\documentclass[lettersize,journal]{IEEEtran}
\usepackage{amsmath,amssymb,mathtools,amsthm,amsfonts,array}
\usepackage{subfigure,graphicx,graphics,color,psfrag}
\usepackage[figuresright]{rotating}
\usepackage{cite,balance,caption}
\allowdisplaybreaks

\usepackage{algorithm,algorithmicx, algpseudocode}
\usepackage{accents,bm,url}
\usepackage[english]{babel}
\usepackage{multirow,enumerate,cases}
\usepackage{stfloats,dsfont,soul,fancyhdr,hhline}
\usepackage{cancel}
\usepackage{soul}
\usepackage{color,xcolor}
\usepackage{dcolumn}
\usepackage{tablefootnote}
\usepackage{bbm}
\usepackage{gensymb}
\usepackage{url}

\newtheorem{theorem}{Theorem}
\newtheorem{lemma}{Lemma}

\newtheorem{proposition}{Proposition}

\newtheoremstyle{remarkcolon}{3pt}{3pt}{\normalfont}{}{\bfseries}{:}{0.5em}{}     

\theoremstyle{remarkcolon}
\newtheorem{remark}{Remark}

\makeatletter
\renewcommand{\proofname}{Proof}
\renewenvironment{proof}[1][\proofname]{%
	\par\pushQED{\qed}\normalfont
	\topsep6\p@\@plus6\p@\relax
	\trivlist
	\item[\hskip\labelsep\itshape #1\@addpunct{:}] 
	\ignorespaces
}{%
	\popQED\endtrivlist\@endpefalse
}
\makeatother

\begin{document}


\title{Flexible-Sector 6DMA Base Station:\\ Modeling and Design }

\author{Yunli Li, \IEEEmembership{Member, IEEE}, Xiaoming Shi, Xiaodan Shao, \IEEEmembership{Member, IEEE}, Jie Xu, \IEEEmembership{Fellow, IEEE}, and Rui Zhang, \IEEEmembership{Fellow, IEEE}
\thanks{An earlier version of this paper was presented in part at the IEEE International Conference on Communications, Glasgow, Scotland, UK, June 2026 [DOI: 10.1109/ICC59461.2026.11587629] \cite{11587629}. }}
 
\markboth{Journal of \LaTeX\ Class Files,~Vol.~14, No.~8, August~2021}%
{Shell \MakeLowercase{\textit{et al.}}: A Sample Article Using IEEEtran.cls for IEEE Journals}

\markboth{}
{Shell \MakeLowercase{\textit{et al.}}: Bare Demo of IEEEtran.cls for IEEE Journals}

\maketitle

\begin{abstract}

Six-dimensional movable antenna (6DMA) technology creates new spatial degrees of freedom by adjusting antenna positions and orientations, but its full implementation may require hardware and control complexity at base stations (BSs). This paper proposes a cost-effective flexible-sector BS architecture, in which antennas can move along a circular track to support common sector rotation and adaptive antenna allocation across sectors according to the spatial user distribution. We consider uplink transmission in a single-cell system and introduce an angular-domain user distribution model to characterize spatial traffic clustering. With zero-forcing reception, we derive the average sum rate as a function of sector rotation and antenna allocation. We further develop a two-step optimization algorithm that combines sector-rotation search with load-aware antenna allocation. The analysis shows that the optimal number of antennas assigned to each sector increases linearly with the number of users in that sector. Under the most favorable angular user distribution, the asymptotic per-user rate gap between the most and least favorable sector-level user distributions approaches $\log_{2}(B)$ bps/Hz as the number of antennas grows, where $B$ is the number of sectors. Numerical results demonstrate higher sum-rate performance over benchmarks for non-uniform user distributions, with robust gains under non-ideal antenna patterns and practical receivers.

\end{abstract}


\section{Introduction}

Multiple-input multiple-output (MIMO) technology has played a pivotal role in wireless communication systems, which introduces new degrees of freedom (DoFs) in the spatial domain to enable significant array gains as well as spatial multiplexing and diversity gains, thereby substantially enhancing wireless communication rate and reliability \cite{lu2014overview}. Recently, wireless communication systems are increasingly shifting toward higher frequency bands such as millimeter-wave (mmWave) \cite{gao2015mmwave} and terahertz (THz) \cite{wan2021terahertz}, with decreasing wavelength which allows for smaller-size antenna elements. This thus facilitates the deployment of MIMO systems at larger scales. In particular, technologies such as massive MIMO \cite{marzetta2015massive,albreem2019massive,6736761,hua2023mimo} and extremely large-scale MIMO \cite{wang2024tutorial,lu2024tutorial,lu2021communicating,zhang2022fast,huang2023sparse,hua2025near} have emerged. These systems leverage extensive antenna arrays to achieve higher array gains, compensate for the severe propagation losses, and more effectively suppress multi-user interference.

Despite such advancements, conventional MIMO systems primarily rely on fixed-antenna deployment at the base station (BS), where both antenna positions and coverage patterns remain static, even in the face of spatial and temporal fluctuations in wireless channels and user distributions. As such, these fixed-antenna BSs possess limited capability in dynamically adjusting beam directions, aperture configurations, and sector boundaries. This inflexibility results in inefficient spectrum usage in areas with low user density and network congestion in typical hotspots like stadiums and major transportation hubs. Consequently, fixed-antenna systems fail to fully exploit the spatial DoFs offered by MIMO, resulting in degraded spectral and energy efficiency, particularly under heterogeneous user distributions, rapidly changing mobility patterns, and highly dynamic channel conditions \cite{gesbert2007shifting,zeng2023task,zheng2021double}.

To tackle these issues, fluid antenna system (FAS) \cite{new2024tutorial} and movable antenna (MA) \cite{zhu2025tutorial} have been proposed as new techniques to enhance the antenna reconfiguration and adaptation capabilities, by dynamically adjusting antenna positions within a given physical aperture to adapt to the channel conditions. As such, they enable more efficient use of antenna DoFs as compared to traditional fixed-position antennas. However, FAS and MA techniques mainly exploit the small-scale channel fading to facilitate constructive/destructive combining of multipath signals, thereby enhancing desired signal strength and eliminating undesired interference for improving communication performance \cite{zhu2024performance,chen2024exploiting,mei2024movable,liu2025near,zheng2024flexible,shi2025joint}. Thus, their practically achievable performance gains critically rely on the availability of accurate channel state information (CSI) and the sufficiently high movement speed of antennas to keep track of fast channel variation.

More recently, six-dimensional movable antenna (6DMA) has been proposed \cite{shao20246d}, which incorporates the new antenna rotational DoF in addition to the antenna position adjustment of FAS/MA. This is achieved through the joint control of both three-dimensional (3D) positions and 3D rotations of antennas or antenna arrays. Notably, 6DMA has been shown to be able to leverage the knowledge of statistical CSI or user spatial distribution to adapt to large-scale channel variations and thereby enhance the multiuser communication performance \cite{shao2025tutorial,shao20256d,shao20256dma}. Moreover, leveraging directional sparsity, channel estimation methods for arbitrary 6DMA positions and rotations were developed in \cite{direct}. Further advancing this architecture, a hierarchically-tunable 6DMA framework was proposed in \cite{hua2025hierarchically}, which pursues a sequential optimization of the arrays' global spherical positioning and individual local orientation tuning, thereby significantly reducing the computational complexity while improving the system performance. Unlike \cite{shi20256DMA}, which studies circular-track antenna movement for cell-free massive MIMO, this work targets sectorized BSs with directional sector gains and jointly optimizes sector rotation and cross-sector antenna allocation. It further develops an angular-domain traffic model, closed-form optimal allocation, and asymptotic rate-gain analysis.

It is also worth noting that existing fixed-antenna BSs widely employ sectorized antenna arrays to achieve effective coverage with high directional gains, termed as sectorized BS \cite{6824752}. Specifically, high-gain directional antennas are employed at the BS to divide its coverage area or cell into orthogonal sectors \cite{dai2013overview}. By concentrating radiated energy into narrow main lobes and suppressing sidelobes, such sector antennas achieve effective spatial isolation between adjacent sectors to suppress co-channel interference \cite{athley2013increased}. This leads to improved signal-to-interference-plus-noise ratio (SINR), denser frequency reuse, and enhanced cell-edge throughput. However, traditional sectorized BSs still lack the capability of flexible antenna movement in terms of positioning and/or rotation offered by MA, FAS, and 6DMA. Therefore, it is appealing to improve sectorized BSs to incorporate such new spatial flexibility for throughput enhancement. Despite their promising potential, the new reconfigurable antenna designs, such as MA/FAS/6DMA, are not compatible with existing sectorized BS architecture, thus requiring fundamental changes of the BS hardware and design.  

Motivated by the above, we propose a sector-constrained 6DMA-inspired BS architecture, termed the flexible-sector BS, which augments conventional sectorized BSs with two slow-timescale reconfiguration degrees of freedom: common sector rotation and cross-sector physical antenna allocation. In particular, antenna modules are constrained to move along a circular track, and their configuration is adapted to long-term angular traffic statistics rather than instantaneous CSI. This design is fundamentally different from conventional dynamic/adaptive sectorization, where sector boundaries or beam directions may be adjusted but the antenna resources associated with each sector are typically fixed. By jointly optimizing sector rotation and physical antenna allocation, the proposed flexible-sector BS redistributes sector-level array gains according to non-uniform user distributions. At the same time, it avoids the hardware and control complexity of general 6DMA architectures with arbitrary antenna positioning and rotation, thereby providing a low-complexity intermediate architecture from fixed-sector BSs to fully reconfigurable 6DMA systems. The main results of this paper are outlined as follows.

\begin{itemize}
	\item We propose an analytical and optimization framework for evaluating the performance of flexible-sector BS systems. Specifically, an angular-domain long-term traffic load model based on a 2D finite homogeneous Poisson point process (FHPPP) is introduced to capture angular user clustering and hotspots. This model greatly facilitates the design of sector rotation and antenna allocation across sectors, while distinguishing from conventional distance-based stochastic models such as those built upon Poisson cluster process (PCP) and Matérn hard-core Poisson point process (PPP) \cite{andrews2016primer,he2021unified,li2024geometric,li2024analysis,haenggi2012stochastic,stoyan1987stochastic}.

	\item Under the angular-domain user distribution model, we study uplink transmission in a single-cell system where multiple single-antenna users communicate with a multi-antenna flexible-sector BS. Each sector employs ZF reception and channel-inversion power control to ensure a common user rate. We then analyze the cell-average sum rate as a function of sector rotation and cross-sector antenna allocation.

	\item Based on the derived rate expression, we formulate a sum-rate maximization problem under total-antenna and per-user minimum-rate constraints. A semi-closed-form solution reveals the load-aware allocation structure, while a globally optimal antenna-allocation algorithm, combined with a one-dimensional rotation angle search, solves the original non-convex problem. The analysis shows that the optimal antenna allocation increases with sector traffic load, and larger gains arise under more uneven user distributions. The asymptotic per-user rate gap between the most and least favorable user distributions approaches $\log_2(B)$ bps/Hz as $N \to \infty$, where $B$ is the number of sectors.

	
	\item Finally, we present numerical results to validate the analysis under representative angular-domain user distributions. The results show that the proposed flexible-sector BS with joint sector rotation and antenna allocation outperforms benchmark schemes with only one adaptation mechanism. The performance gain becomes more pronounced when users are more spatially clustered or when the number of sectors increases. We also evaluate the proposed design under non-ideal sector antenna patterns, including finite-side-lobe and 3GPP-like horizontal patterns, showing that the flexible-sector BS remains effective under practical radiation conditions. 
	
\end{itemize}

The rest of this paper is organized as follows. Section~\ref{sec:systemModel} presents the flexible-sector BS architecture, the angular-domain traffic model, and the resulting channel models. Section~\ref{sec:achRate} characterizes the average rates achievable by users based on the ZF receiver. Section~\ref{sec:probFormulation} formulates the average sum-rate maximization problem. Section~\ref{sec:Soln} presents a two-step algorithm for solving the formulated problem. Section~\ref{sec:effectsPara} investigates the effects of several key factors on the users' sum rate. Section~\ref{sec:numRes} validates the analytical insights and then evaluates the robustness of the proposed flexible-sector BS to angular traffic clustering and non-ideal antenna patterns. Finally, Section~\ref{sec:conclusion} concludes this paper.

\textit{Notations}: Scalars are denoted by lower-case letters, vectors by bold-face lower-case letters, and matrices by bold-face upper-case letters. $\boldsymbol{I}$ and $\boldsymbol{0}$ denote an identity matrix and an all-zero matrix, respectively, with appropriate dimensions. For a square matrix $\boldsymbol{S}$, $\boldsymbol{S}^{-1}$ denotes its inverse (if $\boldsymbol{S}$ is full-rank), and $[\boldsymbol{S}]_{k,k}$ denotes its $k$-th diagonal element. 
For a matrix $\boldsymbol{M}$ of arbitrary size, $\boldsymbol{M}^{\mathsf{H}}$, $\boldsymbol{M}^{\mathsf{T}}$, and $\mathbb{E}[\boldsymbol{M}]$ denote its conjugate transpose, transpose, and statistical expectation, respectively. The distribution of a circularly symmetric complex Gaussian (CSCG) random vector with mean $\boldsymbol{0}$ and covariance matrix $\boldsymbol{\Sigma}$ is denoted by $\mathcal{CN}(\boldsymbol{0},\boldsymbol{\Sigma})$, and $\sim$ stands for “distributed as". $\mathbb{C}^{x\times y}$ denotes the space of ${x\times y}$ complex matrices. 
$\mathbb{C}$ denotes the set of complex numbers. 
$\mathbb{Z}_{++}$ denotes the set of positive integers. 
$\mathbb{Z}_{+}$ denotes the set of non-negative integers. 
$\mathbb{R}_{+}$ denotes the set of non-negative real numbers. 
$\Vert \boldsymbol{x} \Vert$ denotes the Euclidean norm of a complex vector $\boldsymbol{x}$. $\vert \cdot \vert$ denotes the cardinality of a set. $\overline{x}$ represents the mean value of a real number $x$. The operation $+$ is defined as $x^{+}\triangleq\max\{x,0\}$. 
$\binom{\cdot}{\cdot}$ denotes the binomial coefficient. 

\section{System Model}\label{sec:systemModel}

\begin{figure*}[t!]
	\centering
	\subfigure[Flexible-sector BS with $B=3$ sectors.]{
		\begin{minipage}[t]{0.47\textwidth}
			\includegraphics[width=0.9\linewidth]{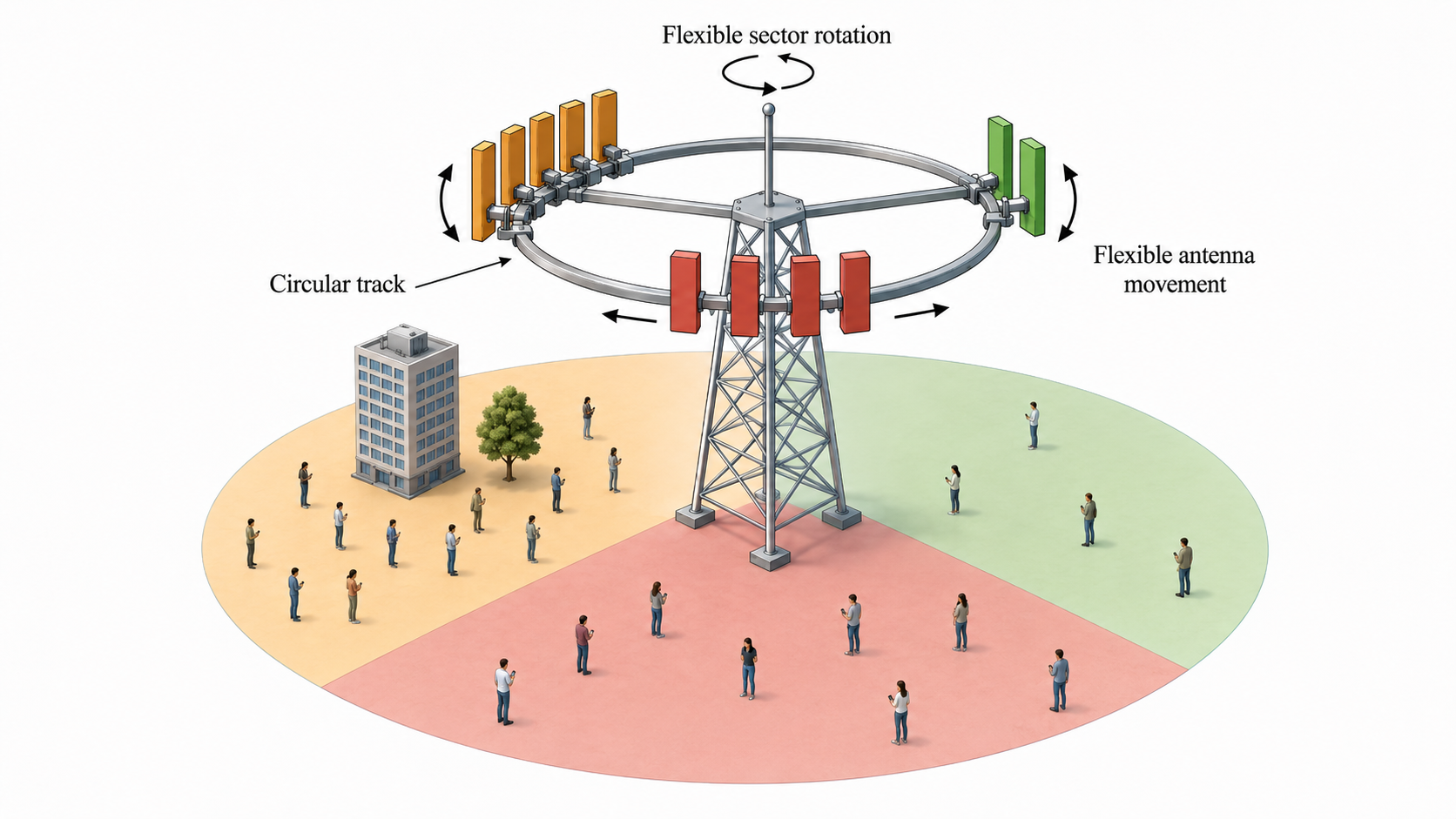}
			\label{fig:systemModel_phy}
		\end{minipage}
	} 
	\subfigure[Angular zone-based user distribution with $c=4$ zones in each sector.]{
		\begin{minipage}[t]{0.47\textwidth}
			\includegraphics[width=0.7\linewidth]{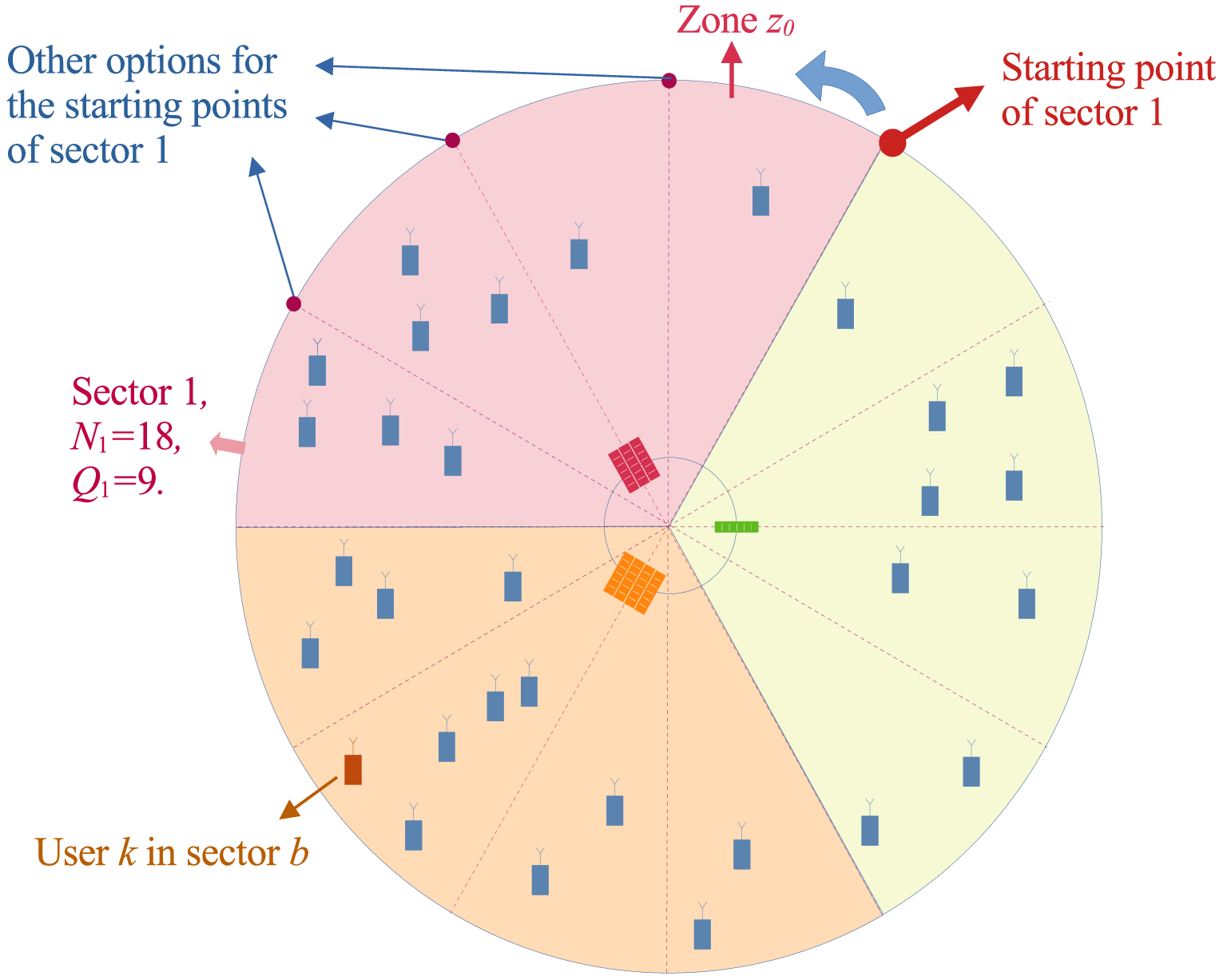}
			\label{fig:systemModel_ana}
		\end{minipage}
	}
	\caption{Illustration of the flexible-sector BS and angular zone-based user distribution.}
	\label{fig:systemModel}
\end{figure*}

In this section, we first present the flexible-sector BS architecture. We then introduce the angular-domain user distribution model and the corresponding channel model. 
 
\vspace{-1em}
\subsection{Flexible-Sector BS}

We consider the uplink transmission in a single-cell system consisting of a multi-antenna BS and multiple single-antenna users, as illustrated in Fig.~\ref{fig:systemModel_phy}. 
We assume that the BS is equipped with $N\geq1$ antennas in total, denoted by set $\mathcal{N}\triangleq\{1,2,\cdots,N\}$. 
The antennas are allocated to $B$ sectors, denoted by set $\mathcal{B}\triangleq\{1,2,\cdots,B\}$, with $B\geq1$. 
Each sector is assumed to have an equal coverage in azimuth angle due to the use of directional antennas, which is given by angle $\Phi\triangleq\frac{2\pi}{B}$. 
Different from conventional BS with a fixed number of antennas in each sector, we propose a flexible-sector BS where the antennas at the BS can be flexibly allocated over different sectors. Specifically, the set of antennas in sector $b\in\mathcal{B}$ is denoted by set $\mathcal{N}_{b},\mathcal{N}_{b}\subset\mathcal{N}$, with $N_{b}\triangleq\vert \mathcal{N}_{b} \vert$ denoting the number of antennas allocated to sector $b$. 
As shown in Fig.~\ref{fig:systemModel_phy}, the antennas are placed along a circular track with radius $d$ at the BS, while those for each sector are placed at the angular center of each sector. In particular, a minimum inter-antenna spacing (guard distance) is enforced to mitigate mutual coupling and spatial correlation among adjacent antennas. The antennas can move freely along the circular track from one sector to another such that the number of antennas in each sector can be flexibly adjusted. 
Moreover, we assume that the angular centers of all antennas in different sectors can rotate flexibly along the circular track by a common azimuth angle, leading to corresponding rotations of orthogonal coverage regions of all sectors in azimuth angle. 
We set the center of the circular track as the reference position of the BS. 
The users are assumed to be randomly distributed on the ground within the circular cell of radius $D$ and centered at the BS location.

To avoid interference across different sectors, we assume that the antennas in each sector have a directional radiation pattern (in the case of $B>1$), which is modeled as a truncated function in azimuth angle.\footnote{When $B=1$ and $\Phi=2\pi$, the introduced system in Fig.~\ref{fig:systemModel} reduces to a single-sector BS with omnidirectional antennas.} 
Specifically, the azimuth beamwidth of each antenna is set to be $\Phi$. 
Thus, the corresponding antenna gain in azimuth angle $\phi$ is given by \cite{balanis2016antenna,lyu2018uav}
\begin{equation}\label{eq:antennaG}
	A(\phi) = \begin{cases}
		\frac{2\pi}{{\Phi}},\quad &~ -\frac{\Phi}{2}\leq\phi\leq\frac{\Phi}{2}, \\ 0, \quad &~\text{otherwise}.
	\end{cases}
\end{equation}
Furthermore, we assume that each user is equipped with an omnidirectional antenna of unit gain.

To facilitate the joint design of antenna allocation among sectors and their common rotation, we propose an angular-domain approach to model the user distribution. 
As shown in Fig.~\ref{fig:systemModel_ana}, the whole cell is equally divided into $Z$ zones in terms of azimuth angle, denoted by set $\mathcal{Z}\triangleq\{1,2,\cdots,Z\}$ with $Z\geq1$. 
For convenience, we assume $Z=c B$, where $c$ is a positive integer. Accordingly, each sector $b\in\mathcal{B}$ consists of $c$ zones, which are denoted by set $\mathcal{Z}_{b}$, expressed as
\begin{equation}\label{eq:zoneInSector}
	\mathcal{Z}_{b} = \left\{
	\operatorname{wrap}_{Z}\!\left(z_0+(b-1)c+\ell\right)
	\right\}_{\ell=0}^{c-1},
\end{equation}
with $\operatorname{wrap}_{Z}(n)\triangleq 1+[(n-1)\bmod Z]$, and $z_{0}\in\left\{1,\cdots, c\right\}$ denoting the index of the first zone in $\mathcal{Z}_{1}$ (assuming that all the zones and sectors are arranged in a counter-clockwise manner). 
Thus, the common rotation of all sectors is controlled by adjusting $z_{0}$, which is referred to as the common rotation index, as shown in Fig.~\ref{fig:systemModel_ana}.

\subsection{Angular-Domain Long-Term Traffic Load Model}

We characterize the long-term angular traffic load by modeling the user locations in each angular zone as a two-dimensional (2D) finite homogeneous Poisson point process (FHPPP) \cite{haenggi2012stochastic}, denoted by $\mathcal{K}_{z}$ with spatial density $\lambda_{z}$, $\forall z\in\mathcal{Z}$. 
Thus, all users in the cell are denoted by set $\mathcal{K}=\bigcup_{z=1}^{Z }\mathcal{K}_{z}$.\footnote{In this work, we consider the channel inversion based power control to equalize the users' effective channel gains with the BS in each zone/sector, and as a result, our proposed design relies on the average number of users in each zone regardless of their locations in the zone. Nevertheless, the results obtained in this paper are also applicable to other (non-uniform) user distributions in each zone provided that the corresponding average numbers of users over angular zones are given.} 
For analytical convenience, we assume that the number of users in each zone $z$ is given, which corresponds to the average number of users derived from the user point process for that zone $z\in\mathcal{Z}$.

Let $K_{z}$ and $K$ denote the average number of users in zone $z$ and that in the whole cell, respectively. 
Based on the FHPPP assumption, we have
\begin{equation}\label{eq:numOfUser_zone}
	{K}_{z} = \frac{1}{2}\psi D^{2} \lambda_{z},\quad  z\in\mathcal{Z},
\end{equation}
and 
\begin{equation}\label{eq:numOfUser}
	{K} = \sum_{z=1}^{Z} {K}_{z} = \sum_{z=1}^{Z} \frac{1}{2}\psi D^{2} \lambda_{z}, 
\end{equation}
where $\psi=\frac{2\pi}{Z}$ is the angular range of each zone. The proposed sector rotation and antenna allocation are optimized based on long-term angular traffic loads, rather than instantaneous user locations.

We assume that each zone contains at least one user, i.e., ${K}_{z}\geq1$, and $K < N$. 
Let $\mathcal{Q}_{b}(z_{0})\triangleq\bigcup_{z\in\mathcal{Z}_{b}}\mathcal{K}_{z}$ denote the set of users located in sector $b$, $b\in\mathcal{B}$, and $Q_{b}(z_{0})$ denote the average number of users in sector $b$, both of which depend on the corresponding zones of sector $b$ that vary with the sectors' common rotation (i.e., $z_{0}$). 
Thus, $Q_{b}(z_{0})$ is given by
\begin{equation}\label{eq:numOfUser_sector}
	{Q}_{b}(z_{0}) = \sum_{z\in\mathcal{Z}_{b}}{K}_{z} = \sum_{z=z_{0}+(b-1)c}^{z_{0}+bc-1} \frac{1}{2}\psi D^{2} \lambda_{z},  \quad  b\in\mathcal{B}.
\end{equation}
For convenience, we assume in this paper that ${Q}_{b}(z_{0}) \geq 1, \forall b\in\mathcal{B}$, regardless of $z_0$, which is practically valid in general. We use long-term average user numbers as deterministic integer loads for ZF design and antenna allocation. 


\subsection{Channel Model}

Last, we present the channel model between users and the BS. 
For the ideal sector pattern in \eqref{eq:antennaG}, the channel vector from user $k$ in sector $b$ to the BS is given by
\begin{equation}\label{eq:zone_Ch}
	\boldsymbol{h}_{k} = \sqrt{  B  \zeta_{k}}  \boldsymbol{g}_{k} , \quad k\in\mathcal{Q}_{b},
\end{equation}	
where $B$ accounts for the directional antenna gain in each sector, $\zeta_{k}$ denotes the average channel power gain from user $k$ to the reference point of the BS, and $\boldsymbol{g}_{k}\sim\mathcal{CN}(\boldsymbol{0},\boldsymbol{I}_{N_{b}})$ denotes the complex channel vector characterizing the small-scale Rayleigh fading between user $k$ and the antennas in sector $b$ of BS, $\mathcal{N}_{b}$. Let $\boldsymbol{G}_{b}=[\boldsymbol{g}_{1},\cdots,\boldsymbol{g}_{Q_{b}}]\in\mathbb{C}^{N_{b}\times Q_{b}}$ represent the small-scale fading channel matrix from all users in $\mathcal{Q}_{b}$ to the BS's receive antennas in $\mathcal{N}_{b}$.

The received signal at the BS in sector $b$ is given by
\begin{equation}
	\boldsymbol{y}_{b}  =  \sum_{k\in\mathcal{Q}_{b}}  \boldsymbol{h}_{k} x_{k} + \boldsymbol{\varsigma}_{b}  = \boldsymbol{H}_{b}\boldsymbol{x}_{b} + \boldsymbol{\varsigma}_{b} , \quad b\in\mathcal{B},
\end{equation}
where $x_{k}\in\mathbb{C}$ represents the transmitted information signal of user $k$ with average power $P_{k}$, 
$\boldsymbol{H}_{b} = [\boldsymbol{h}_{1},\cdots,\boldsymbol{h}_{Q_{b}}]\in\mathbb{C}^{N_{b}\times Q_{b}}$ represents the channel matrix from all users in $\mathcal{Q}_{b}$ to the BS's receive antennas in $\mathcal{N}_{b}$,\footnote{For notational convenience, we assume that the users in sector $b$ are indexed from $1$ to $Q_{b}$, and write $Q_b(z_{0})$ as $Q_b$ in short.}
$\boldsymbol{x}_{b}=[x_{1},\cdots,x_{Q_{b}}]^T\in\mathbb{C}^{Q_{b}\times1}$, and $\boldsymbol{\varsigma}_{b}\sim\mathcal{CN}(\boldsymbol{0},\delta^2\boldsymbol{I}_{N_{b}})$ denotes the additive white Gaussian noise (AWGN) at the BS receiver, with $\delta^2$ denoting the noise power. To achieve fair rate allocation among users, we consider the channel inversion power control at the users, such that the user transmit power is inversely proportional to the average channel power gain \cite{goldsmith2005wireless} with $P_{k}=\frac{P_{0}}{\zeta_{k}}$, where $P_{0}$ denotes the common power for all users in the cell. This ensures that the average received power is identical for users in the same sector.

\subsection{Practical Radiation Pattern and Antenna Positioning}
\begin{figure}[t!]
	\centering 
	\begin{minipage}[t]{0.47\linewidth}
		\includegraphics[width=\linewidth]{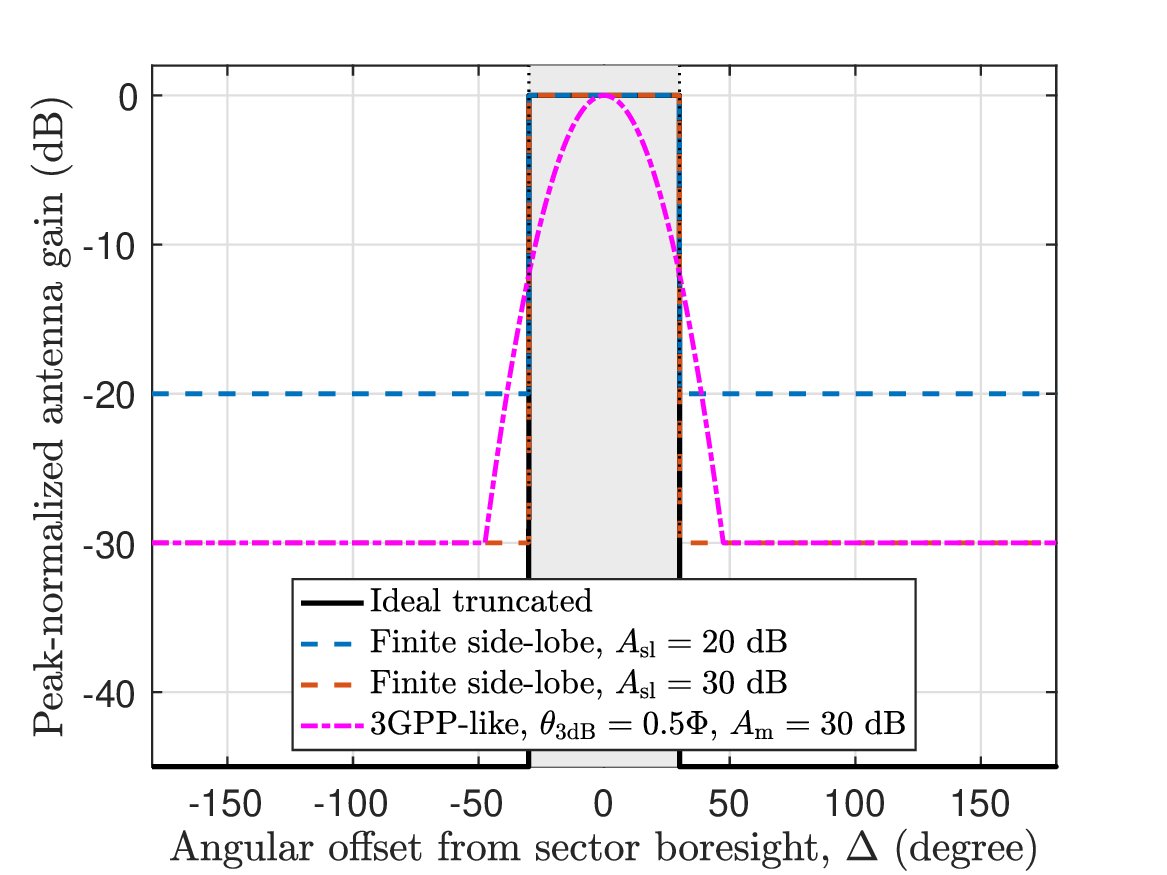} 
		\par\vspace{4pt}  
		\small (a) Cartesian representation versus the angular offset from the sector boresight.
		\label{fig:antennaPattern}
	\end{minipage} 
	\begin{minipage}[t]{0.47\linewidth}
		\includegraphics[width=0.8\linewidth]{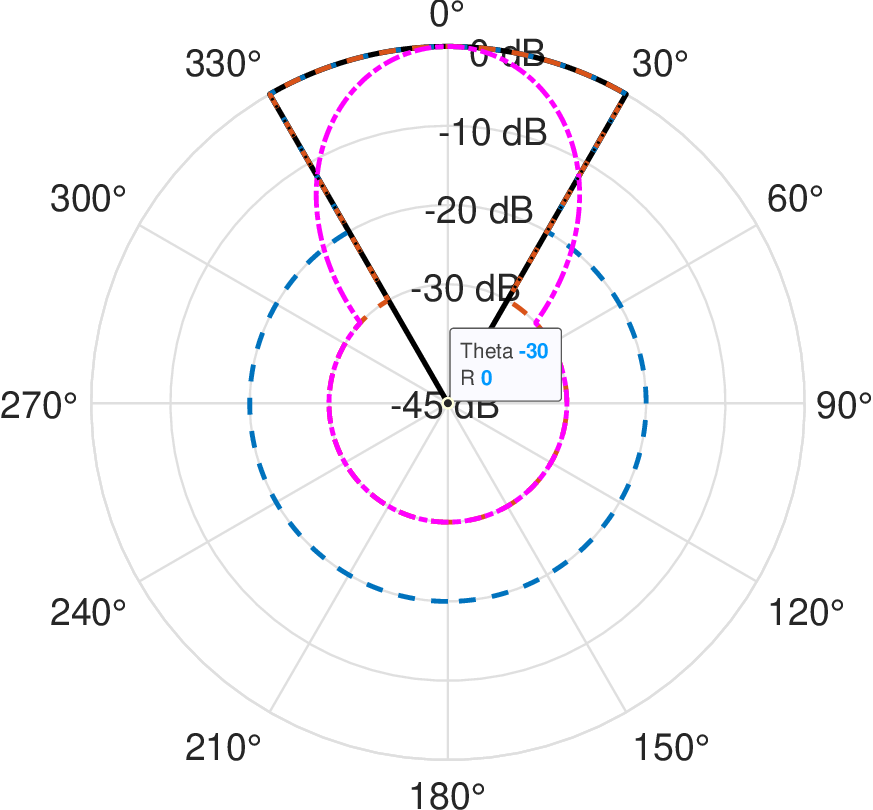}
		\par\vspace{4pt}  
		\small (b) Polar representation. 
		\label{fig:antennaPattern_polar}
	\end{minipage} 
	\caption{Antenna pattern illustration.}
	\label{fig:antennaPattern0}
\end{figure} 
The ideal sector pattern in \eqref{eq:antennaG} is adopted for analytical tractability, since it leads to orthogonal sector coverage and enables closed-form rate characterization. In practical sector antennas, however, the radiation power outside the main sector cannot be perfectly suppressed due to finite side-lobe levels, front-to-back leakage, and sector-boundary overlap. To evaluate the robustness of the proposed flexible-sector BS under practical radiation patterns, we also consider the following finite-side-lobe sector antenna model. 

Let $\varphi_k$ denote the azimuth angle of user $k$ with respect to the BS, and $\vartheta_b$ denote the boresight direction of sector $b$. The angular offset between user $k$ and sector $b$ is defined as
	\begin{equation}
		\Delta_{b,k} ={\rm wrap}_{\pi}\left(\varphi_k-\vartheta_b \right),
	\end{equation} 
	where ${\rm wrap}_{\pi}(\cdot)$ maps an angle into $[-\pi,\pi]$. The practical finite-side-lobe antenna gain is modeled as
	\begin{equation}	\label{eq:ASL_antenna}
		A_\eta(\Delta_{b,k})=
		\begin{cases}
			G_{\rm m}, & |\Delta_{b,k}|\leq \frac{\Phi}{2},\\
			\eta G_{\rm m}, & \text{otherwise},
		\end{cases}
	\end{equation}
	where $\eta=10^{-A_{\rm sl}/10}$ denotes the side-lobe leakage ratio with $A_{\rm sl}$ being the side-lobe attenuation in dB. The main-lobe gain $G_{\rm m}$ is chosen as
	\begin{equation}
		G_{\rm m}=\frac{1}{\frac{1}{B}+\eta\left(1-\frac{1}{B}\right)},
	\end{equation}
	such that the average antenna gain over the azimuth domain is normalized to one. When $\eta=0$, the above model reduces to the ideal sector pattern in \eqref{eq:antennaG}. Therefore, \eqref{eq:antennaG} can be regarded as an ideal upper-bound model, while $A_\eta(\Delta)$ captures practical inter-sector leakage caused by finite side lobes.

In addition, a 3GPP-like horizontal pattern is plotted for reference, whose normalized dB-domain gain is given by
	\begin{equation} \label{eq:3gpp_antenna}
		A_{{\rm 3GPP},{\rm dB}}(\Delta) = -\min\left\{	12\left(\frac{\Delta}{\theta_{3{\rm dB}}}\right)^2, A_m\right\}.
	\end{equation}

For practical evaluation with finite side lobes, the channel from user $k$ to the antennas in sector $b$ is modeled as
	\begin{equation}
		\mathbf {\tilde h}_{b,k}=	\sqrt{\zeta_k A_\eta(\Delta_{b,k})}\mathbf g_{b,k},\quad b\in\mathcal B,\ k\in\mathcal K,
	\end{equation}
	where $\mathbf g_{b,k}\sim\mathcal{CN}(\mathbf 0,\mathbf I_{N_b})$. Under the ideal case with $\eta=0$, $\mathbf {\tilde h}_{b,k}$ is nonzero only when user $k$ lies within the angular coverage of sector $b$, and the above model reduces to \eqref{eq:zone_Ch}. Under the non-ideal sector pattern, signals from users outside sector \(b\) may leak into the receive antennas of sector \(b\). Therefore, the received signal at sector \(b\) is evaluated as
	\begin{equation}
		\tilde{\boldsymbol{y}}_b
		=
		\sum_{k\in\mathcal K} 
		\mathbf {\tilde h}_{b,k}x_k
		+
		\boldsymbol\varsigma_b.
	\end{equation}

Moreover, the proposed flexible-sector BS operates on a slow timescale. Antenna modules mounted on a motorized circular rail are repositioned according to long-term angular traffic statistics to realize sector rotation and cross-sector antenna allocation. After reconfiguration, RF calibration and channel estimation are performed, and the antenna configuration remains fixed during data transmission, while ZF combining is updated using instantaneous uplink channel state information (CSI).

\section{Average Achievable Rate}\label{sec:achRate}

In this section, we first derive the achievable rate of user $k$ in sector $b$, and then obtain the sum rate achievable for all users served by the BS. In order for ZF combining to be feasible, it is assumed that the number of antennas $N_{b}$ in sector $b$ is always no smaller than that of users $Q_{b}$, i.e., $N_{b} \geq Q_{b}$. 
We assume that the channel coefficients between all users and their serving antennas in each sector are perfectly known at the BS. 
For the ideal antenna pattern, the received signals from all users in sector $b$ at its antennas are linearly combined, yielding
\begin{equation}
	\widehat{\boldsymbol{x}}_{b}  = \boldsymbol{W}_{b} ^{\mathsf{H}}\boldsymbol{y}_{b} , \quad b\in\mathcal{B},
\end{equation}
where $\boldsymbol{W}_{b} =[\boldsymbol{w}_{1},\cdots,\boldsymbol{w}_{Q_{b}}]$, with $\boldsymbol{w}_{k}\in\mathbb{C}^{N_{b}\times1},k\in\mathcal{Q}_{b},$ denoting the combining vector for user $k$ with $\Vert\boldsymbol{w}_{k}\Vert=1$.

In this work, for simplicity, we consider the ZF combining at each sector to completely eliminate the inter-user interference, i.e., $\boldsymbol{w}_{k}^{\mathsf{H}}\boldsymbol{g}_{j}=0,\forall j,k \in Q_b, j\neq k$.\footnote{The ZF receiver is adopted to obtain tractable results and useful insights. Our results are extendable to other types of linear receivers, such as linear minimum mean-square error (LMMSE) receiver and regularized ZF receiver \cite{zhang2025fluid} (see Section \ref{sec:numRes} for more details).  }
Specifically, we define 
\begin{equation}
	\overline{\boldsymbol{W}}_{b} = [\overline{\boldsymbol{w}}_{1},\cdots,\overline{\boldsymbol{w}}_{Q_{b}}] = \boldsymbol{G}_{b}(\boldsymbol{G}_{b}^{\mathsf{H}}\boldsymbol{G}_{b})^{-1}, \quad b\in\mathcal{B}.
\end{equation}
Thus, the ZF combining vector for user $k$ is given by
\begin{equation}\label{eq:ZF_vector}
	\boldsymbol{w}_{k} = \frac{\overline{\boldsymbol{w}}_{k}}{\Vert \overline{\boldsymbol{w}}_{k} \Vert} , \quad  k\in\mathcal{Q}_{b}.
\end{equation}

Accordingly, the signal-to-noise ratio (SNR) of user $k$ in sector $b$ is given by \cite{liu2019comp}
\begin{equation}
	\gamma_{b,k} =  \frac{P_{k}\vert \boldsymbol{w}_{k}^{\mathsf{H}}\boldsymbol{h}_{k}  \vert^{2} }{\delta^2} = \frac{P_{0}B}{[(\boldsymbol{G}_{b}^{\mathsf{H}}\boldsymbol{G}_{b})^{-1}]_{k,k}\delta^{2}}, \quad  k\in\mathcal{Q}_{b}.
\end{equation}
The average achievable rate of user $k$ in sector $b$ is given by
\begin{equation}
	{r}_{b,k} = \mathbb{E}\left[ \log_{2}\left( 1+ \gamma_{b,k} \right)\right], \quad  k\in\mathcal{Q}_{b} , 
\end{equation}
where the expectation is taken over the randomness of the small-scale Rayleigh fading channel coefficients.

Since the channel-inversion power control is applied within each sector, the average rates achievable for all users in the same sector are identical. 
Let ${r}_{b},b\in\mathcal{B}$, denote the (identical) achievable rate of each user in sector $b$, with ${r}_{b}=r_{b,k}, \forall k\in\mathcal{Q}_{b}$.
We have the following results to characterize the upper and lower bounds for rate $r_b$, denoted by ${r}_{b}^{(\rm u)}$ and ${r}_{b}^{(\rm l)}$, respectively.

\begin{lemma}\label{lem:bounds}
	The average achievable rate per user in sector $b$ satisfies the following inequalities:
	\begin{equation}\label{eq:bounds}
		{r}_{b}^{(\rm l)} \leq {r}_{b} \leq {r}_{b}^{(\rm u)}, \quad b\in\mathcal{B},
	\end{equation} 
	where 
	\begin{equation}\label{eq:upperB_r_b}
		{r}_{b}^{(\rm u)} =  \log_{2}\left( 1+  \mathbb{E}\left[ \frac{P_{k}\vert \boldsymbol{w}_{k}^{\mathsf{H}}\boldsymbol{h}_{k}  \vert^{2} }{\delta^2}\right] \right), \quad k\in\mathcal{Q}_{b},
	\end{equation}
	and	
	\begin{equation}\label{eq:lowerB_r_b}
		{r}_{b}^{(\rm l)} =  \log_{2}\left( 1+ \frac{P_{0}B}{\mathbb{E}\left[[(\boldsymbol{G}^{\mathsf{H}}_{b}\boldsymbol{G}_{b})^{-1}]_{k ,k }\right]\delta^{2}}\right) , \quad k\in\mathcal{Q}_{b}.
	\end{equation}
	More specifically, ${r}_{b}^{(\rm u)}$ in \eqref{eq:upperB_r_b} can be expressed as a function of the common rotation index, $z_{0}$, and the number of antennas allocated to sector $b$, $N_{b}$, which is given by
	\begin{equation}\label{eq:R_k_upper}
		{r}_{b}^{(\rm u)}(N_{b},z_{0}) = \log_{2}\left[ 1+B\gamma_{0} (N_{b}-Q_{b}(z_{0})+1)^{+}\right] , \quad b\in\mathcal{B},
	\end{equation}
	where $\gamma_{0}=\frac{P_{0}}{\delta^{2}}$; and ${r}_{b}^{(\rm l)}$ in \eqref{eq:lowerB_r_b} can be expressed as 
	\begin{equation}\label{eq:R_k_lower}
		{r}_{b}^{(\rm l)} (N_{b},z_{0}) = \log_{2}\left[ 1+ B \gamma_{0} (N_{b}-Q_{b}(z_{0}))^{+}\right] , \quad b\in\mathcal{B}.
	\end{equation}
\end{lemma}
\begin{proof}
	See Appendix~\ref{append:Lemma1}. 
\end{proof}

\begin{figure}[t!] 
	\centering
	\includegraphics[width=0.6\linewidth]{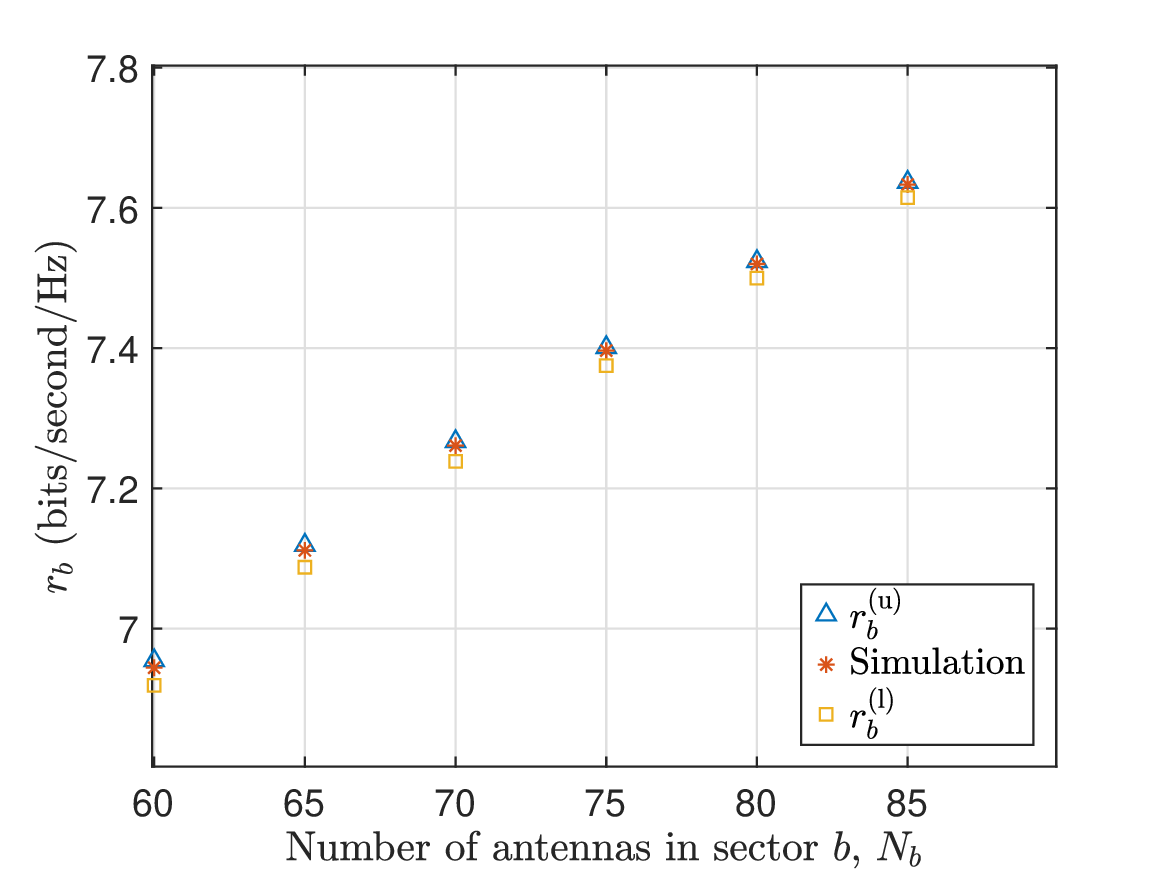} 
	\caption{Simulated user achievable rate versus the number of antennas in sector $b$.}
	\label{fig:bounds}
\end{figure}

From Lemma~\ref{lem:bounds}, it follows that the upper and lower bounds for the achievable rate of users in sector $b$, given in \eqref{eq:R_k_upper} and \eqref{eq:R_k_lower} respectively, are explicitly expressed in terms of the numbers of antennas and users in the sector. Moreover, it is observed from Lemma~\ref{lem:bounds} that the difference between the upper and lower bounds is only due to that between $[N_{b}-Q_{b}(z_{0})+1]^{+}$ and $[N_{b}-Q_{b}(z_{0})]^{+}$, and thus these bounds given in Lemma~\ref{lem:bounds} are generally very tight \cite{liu2019comp}, as shown in Fig.~\ref{fig:bounds}, especially when $(N_{b}-Q_{b})$ is much larger than $1$, which usually holds in practical massive MIMO systems. As such, we adopt the lower bound as the average achievable rate per user for sector $b$ for system design and optimization in the rest of this paper for obtaining the worst-case achievable rate. 
Accordingly, the lower bound of the sum rate of all users in sector $b$ is given by

\begin{equation}\label{eq:R_b}
	R_{b}(N_{b},z_{0}) \triangleq   {Q}_{b}(z_{0}) {r}_{b}^{(\rm l)}(N_{b},z_{0}), \quad b\in\mathcal{B}.
\end{equation} 
Finally, the lower bound of the sum rate of all users in all sectors of the cell is given by
\begin{equation}\label{eq:sumR_obj}
	\begin{split}
		R(\boldsymbol{n},z_{0})  = &~  \sum_{b=1}^{B} R_{b}(N_{b},z_{0}) \\ = &~ \sum_{b=1}^{B} {Q}_{b}(z_{0}) \log_{2}\left[ 1+ B \gamma_{0}(N_{b}-Q_{b}(z_{0}))^{+}\right] ,
	\end{split}
\end{equation}  
where $\boldsymbol{n}=[N_{1},\cdots,N_{B}]^{\mathsf{T}}\in\mathbb{Z}_{+}^{B}$ denotes the collective vector indicating the number of antennas allocated to each sector.

\section{Problem Formulation}\label{sec:probFormulation}

In this paper, we aim to maximize a lower bound on the average sum rate subject to per-user rate constraints served by the flexible-sector BS through the joint optimization of the antenna allocation $\boldsymbol{n}$ over different sectors and the common rotation index $z_{0}$ of all sectors. In the following, we use the lower bound given in \eqref{eq:sumR_obj} to represent the sum-rate of all users for convenience. The sum-rate maximization problem is thus formulated as  
\begin{subequations}
	\begin{alignat}{2}
		({\rm P}1): \quad \max_{\boldsymbol{n}\in\mathbb{Z}_{+}^{B},z_{0}\in\mathbb{Z}_{++}} ~ & ~~  R(\boldsymbol{n},z_{0})   \notag	\\ 
		\mathrm {s.t.}  &~~ \sum_{b=1}^{B}N_{b} \leq N,  \label{con:1-1}	\\ &~~ r_{b}^{(l)}(N_{b},z_{0}) \geq  \overline{r}, \forall b\in\mathcal{B}, \label{con:1-2} \\ &~~ 1\leq z_{0}\leq c, \label{con:1-3} 
	\end{alignat}
\end{subequations}
where \eqref{con:1-1} denotes the total antenna budget at the BS, and \eqref{con:1-2} represents the minimum rate constraints for all users with 
$\overline{r}>0$ denoting the minimum rate required for all users. 
$({\rm P}1)$ is an integer programming problem and can be optimally solved by the exhaustive search over $\boldsymbol{n}$ and $z_{0}$ in their joint discrete set. However, this approach fails to reveal useful insights into the optimal solution, and is computationally prohibitive (with complexity in the order of $O\left(c\binom{N-1}{B-1}\right)$) when $N$ and/or $B$ is large. Thus, in this paper, we propose an alternative approach to solve $({\rm P}1)$, by relaxing the discrete variables $N_{b},\forall b\in\mathcal{B},$ into continuous ones, as detailed in Section \ref{sec:Soln}. 

Before we proceed to solve $({\rm P}1)$, we note that to ensure the feasibility of problem $({\rm P}1)$, from \eqref{eq:R_k_lower} and \eqref{con:1-2} the following constraints need to be satisfied: 
\begin{equation}\label{eq:rho_bar}
	N_{b} \geq Q_{b}(z_{0}) + \frac{2^{\overline{r}}-1}{B\gamma_{0}} \triangleq  N_{b,\min} , \quad \forall   b\in\mathcal{B},
\end{equation}
which, together with the constraint \eqref{con:1-1} and the fact that $\sum_{b=1}^{B} Q_{b}(z_{0}) = K$, leads to the following constraint over $\overline{r}$:
\begin{equation}\label{eq:con_tau}
	\overline{r} \leq \log_{2}\left[1+\gamma_{0}(N-K)\right] \triangleq \overline{r}_{\max}.
\end{equation}
In the sequel of this paper, we replace the constraint \eqref{con:1-2} of $({\rm P}1)$ with its equivalent one in \eqref{eq:rho_bar}.


\section{Proposed Solution to Problem $({\rm P}1)$}\label{sec:Soln}

Since the feasible set of the common rotation of sectors $z_{0}$ consists of a finite number of discrete values, we address problem $({\rm P}1)$ by adopting a two-step approach. 
First, we determine the antenna allocation policy for $\boldsymbol{n}$ with any given $z_{0}$. 
Then, we perform a discrete search over the feasible set of $z_{0}$ to determine the overall solution to $({\rm P}1)$.

\vspace{-1em}
\subsection{Antenna Allocation Optimization with Given Common  Sector Rotation }
With given common rotation index $z_{0}$, the optimization of antenna allocation $\boldsymbol{n}$ in problem $({\rm P}1)$ is simplified as
\begin{subequations}
	\begin{alignat}{2}
		({\rm P}2):\quad  \max_{\boldsymbol{n}\in\mathbb{Z}_{+}^{B}} ~ & ~~  \sum_{b=1}^{B} Q_{b}\log_{2}\left[ 1+ B\gamma_{0}(N_{b}-Q_{b}) \right] \notag	\\ 
		\mathrm {s.t.}   &~~ \sum_{b=1}^{B}N_{b} \leq  N,  \label{con:2-1}	\\ &~~ 	N_{b} \geq Q_{b} + \frac{2^{\overline{r}}-1}{B\gamma_{0}}, \quad \forall   b\in\mathcal{B}.  \label{con:2-2} 
	\end{alignat}
\end{subequations}
Note that in the objective function of $({\rm P}2)$ we omit $z_0$ in $Q_b(z_0)$ since $z_0$ is fixed. In addition, we have removed the `$+$' operation in the objective function of $({\rm P}2)$ due to constraint \eqref{con:2-2}. 

Next, we relax $N_{b},\forall b\in\mathcal{B},$ to continuous variables. 
Thus, problem $({\rm P}2)$ is reformulated as
\begin{subequations}
	\begin{alignat}{2}
		({\rm P}3):\quad  \max_{\boldsymbol{n}\in\mathbb{R}_{+}^{B}} ~ & ~~  \sum_{b=1}^{B} Q_{b}\log_{2}\left[ 1+ B\gamma_{0}(N_{b}-Q_{b}) \right]  \notag	\\ 
		\mathrm {s.t.}  &~~  \sum_{b=1}^{B}N_{b}\leq  N,  \label{con:3-1}		\\ &~~  N_{b} \geq N_{b,\min} ,\quad \forall b\in\mathcal{B}. \label{con:3-2}
	\end{alignat}
\end{subequations}
Thus, problem $({\rm P}2)$ is transformed from a non-convex integer program to a convex optimization problem over continuous $N_b$'s.

It can be verified that problem $({\rm P}3)$ is a convex optimization problem, since the objective function is concave and the constraints are linear.  
\begin{proposition}\label{prop:closedForm}
	The optimal semi-closed-form solution to $({\rm P}3)$ is given by
	\begin{equation}\label{sol:antennaAllo}
		N_{b}^{\star} = \max\left\{{Q_{b}}\left( 1+\frac{1}{\nu{\ln2}} \right) - \frac{1}{B\gamma_{0}},N_{b,\min}\right\}, \quad b\in\mathcal{B}, 
	\end{equation}
	where $ {\nu}$ is a unique constant that satisfies the following equation: 
	\begin{equation}\label{eq:nu_constraint}
		\sum_{b=1}^{B} \max\left\{ {Q_{b}}\left( 1+\frac{1}{\nu{\ln2}} \right) - \frac{1}{B\gamma_{0}},N_{b,\min}\right\}=N.
	\end{equation}
\end{proposition}
\begin{proof}
	See Appendix~\ref{append:Lagrange}.
\end{proof}

The value of $\nu$ can be numerically found by a simple bisection search based on \eqref{eq:nu_constraint}. From \eqref{eq:rho_bar} and \eqref{sol:antennaAllo}, it is worth noting that $N_b^{\star}$ is a monotonically increasing function of $Q_b$, which is intuitively expected since more antennas need to be allocated to sector $b$ when there are more users in it for achieving higher rates. An illustration for $N_b^{\star}$ given in \eqref{sol:antennaAllo} is provided in Fig.~\ref{fig:water-fillingSoln}.

\begin{figure}[t!] 
	\centering
	\includegraphics[width=0.6\linewidth]{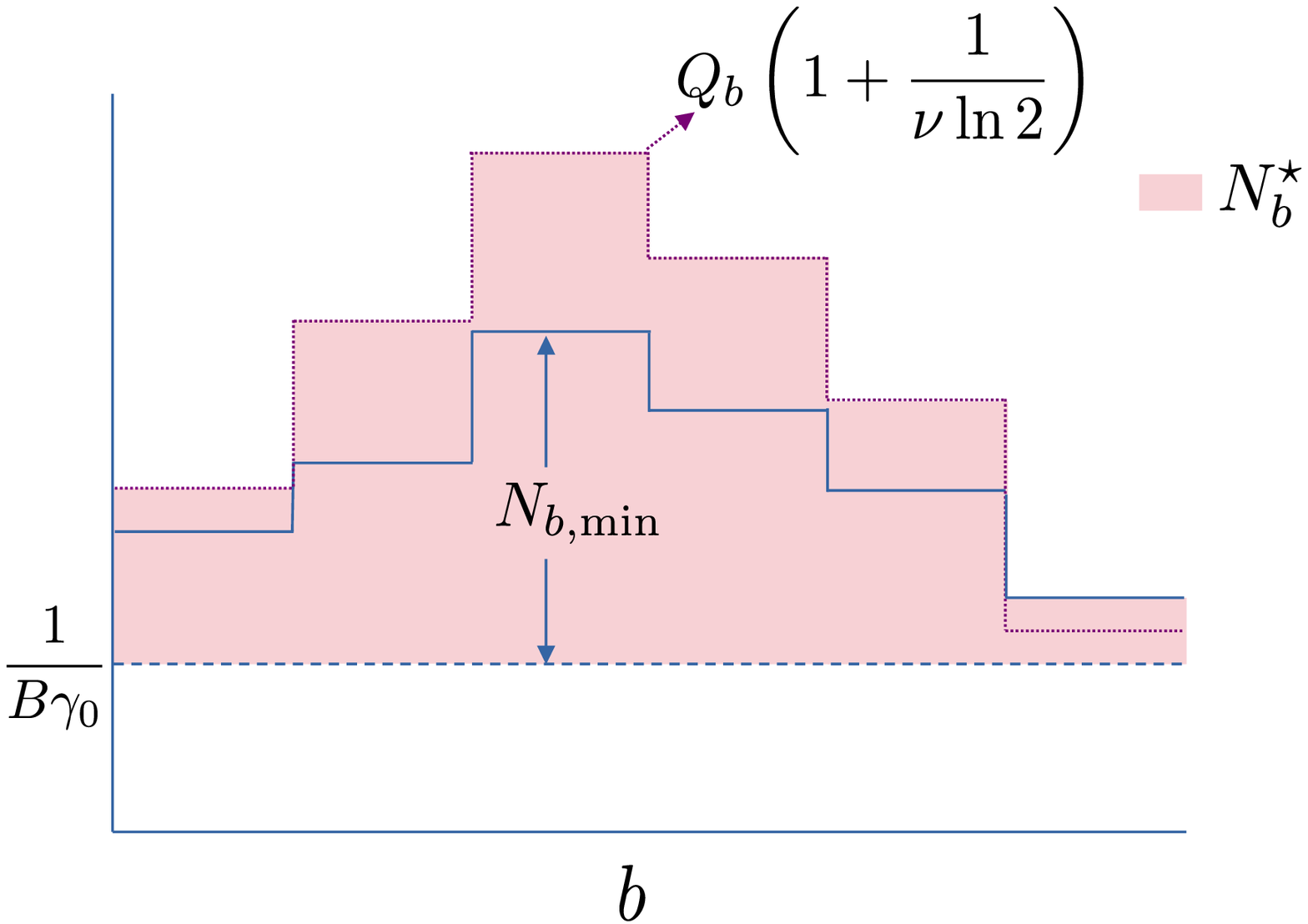} 
	\caption{Illustration of the solution given in \eqref{sol:antennaAllo}.}
	\label{fig:water-fillingSoln}
\end{figure} 

Although Proposition \ref{prop:closedForm} provides useful structural insights into the load-aware antenna allocation, its solution is generally non-integer. Instead of rounding the relaxed solution, we next solve the original integer allocation problem $({\rm P}2)$ directly. By exploiting the diminishing marginal-gain property of the sector-wise rate functions, we develop a greedy allocation procedure and prove that it achieves the globally optimal integer solution for any given common rotation index $z_0$. Instead of directly applying the floor operation, which may violate the minimum-rate constraint and leave some antennas unused, we construct a feasible integer antenna allocation as follows. For a given common rotation index $z_0$, the minimum integer number of antennas required by sector $b$ is given by
\begin{equation}
	L_b(z_0)=\left\lceil Q_b(z_0)+\frac{2^{\bar r}-1}{B\gamma_0}	\right\rceil, \quad b\in\mathcal B .
\end{equation}
If $\sum_{b=1}^{B}L_b(z_0)>N$, then the corresponding $z_0$ is infeasible. Otherwise, we initialize $\hat N_b(z_0)=L_b(z_0)$, $\forall b\in\mathcal B$, and allocate the remaining antennas one by one. The marginal sum-rate gain of assigning one additional antenna to sector $b$ is given by
\begin{equation}\label{eq:opt_nb_}
	\Delta_b(\hat N_b)
	=
	Q_b(z_0)
	\log_2
	\frac{
		1+B\gamma_0\left(\hat N_b+1-Q_b(z_0)\right)
	}{
		1+B\gamma_0\left(\hat N_b-Q_b(z_0)\right)
	}.
\end{equation}
The marginal gain in \eqref{eq:opt_nb_} satisfies a diminishing-return property as shown below.
\begin{proposition}
	The marginal rate gain $\Delta_b(N_b)$ is strictly decreasing with $N_b$.
\end{proposition}

\textit{Proof:} Let $a=B\gamma_0>0$. Treating $N_b$ as a continuous variable, we have
\begin{equation}
	\frac{\partial \Delta_b(N_b)}{\partial N_b}
	=
	-\frac{Q_ba^2}
	{\ln 2\,[1+a(N_b-Q_b)][1+a(N_b-Q_b+1)]}<0. \notag
\end{equation}
Therefore, $\Delta_b(N_b+1)<\Delta_b(N_b)$, which proves the diminishing marginal-gain property. \hfill $\blacksquare$

Thus, at each step, one antenna is assigned to the sector with the largest marginal gain, i.e.,
\begin{equation}
	b^\star
	=
	\arg\max_{b\in\mathcal B}
	\Delta_b(\hat N_b),
	\quad
	\hat N_{b^\star}(z_0)
	\leftarrow
	\hat N_{b^\star}(z_0)+1.
\end{equation}
This process continues until $\sum_{b=1}^{B}\hat N_b(z_0)=N$. The resulting integer antenna allocation is denoted by
\begin{equation} \label{eq:N_b_star}
	N_b^{\star}(z_0)=\hat N_b(z_0),\quad \forall b\in\mathcal B.
\end{equation} 
\begin{remark}\label{rmk:opt_soln}
	For any given feasible common rotation index $z_0$, the marginal-gain-based antenna allocation in \eqref{eq:N_b_star} yields a globally optimal integer solution to problem $({\rm P}2)$.
\end{remark}
According to Remark \ref{rmk:opt_soln}, for every feasible $z_0$, Algorithm \ref{alg:soln} obtains the globally optimal integer antenna allocation. Since the feasible set of $z_0$ is finite and Algorithm \ref{alg:soln} exhaustively evaluates all $z_0\in\{1,\cdots,c\}$, selecting the rotation index yielding the largest objective value gives the globally optimal solution to $({\rm P}1)$.

\vspace{-1em}

\begin{algorithm}[t]
	\caption{Proposed solution for problem $({\rm P}1)$}
	\label{alg:soln}
	\begin{algorithmic}[1]
		\Require Minimum rate requirement $\bar r$ and total antenna number $N$.
		\For{$z_0 = 1,\ldots,c$}
		\State Calculate $Q_b(z_0)$, $\forall b\in\mathcal B$, using \eqref{eq:numOfUser_sector}.
		\State Set $L_b(z_0)=\left\lceil Q_b(z_0)+\frac{2^{\bar r}-1}{B\gamma_0}\right\rceil$, $\forall b\in\mathcal B$.
		\If{$\sum_{b=1}^{B}L_b(z_0)>N$}
		\State Mark $z_0$ as infeasible and continue.
		\EndIf
		\State Initialize $\hat N_b(z_0)=L_b(z_0)$, $\forall b\in\mathcal B$.
		\While{$\sum_{b=1}^{B}\hat N_b(z_0)<N$}
		\State Compute
		$\Delta_b(\hat N_b)=Q_b(z_0)\log_2\frac{1+B\gamma_0(\hat N_b+1-Q_b(z_0))}{1+B\gamma_0(\hat N_b-Q_b(z_0))}$,  $\forall b\in\mathcal B $.
		\State $b^\star=\arg\max_{b\in\mathcal B}\Delta_b(\hat N_b)$.
		\State $\hat N_{b^\star}(z_0)\leftarrow \hat N_{b^\star}(z_0)+1$.
		\EndWhile
		\State Set $N_b^\star(z_0)=\hat N_b(z_0)$, $\forall b\in\mathcal B$.
		\State Calculate $R(z_0)$ using \eqref{eq:sumR_objZ}.
		\EndFor
		\State $z_0^\star=\arg\max_{z_0\in\{1,\ldots,c\}}R(z_0)$.
		\State\Return $z_0^\star$ and $\mathbf n^\star(z_0^\star)$.
	\end{algorithmic}
\end{algorithm}


\subsection{Optimization of Common Rotation Index }\label{sec:optSecRot}

By substituting $N_{b}^{\star} (z_{0})$ given in \eqref{eq:N_b_star} into the objective function of problem $({\rm P}1)$, we obtain the users' maximum sum rate with given $z_0$ as 
\begin{equation}\label{eq:sumR_objZ}
	R(z_{0})  = \sum_{b=1}^{B} {Q}_{b}(z_{0}) \log_{2}\left\{1+ B \gamma_{0}[{N_{b}}^{\star}(z_{0})-Q_{b}(z_{0})] \right\}  .
\end{equation}  
Then, the optimal common rotation index can be obtained by a one-dimensional exhaustive search over $z_0\in\{1,\ldots,c\}$, and is denoted by $z_0^{\star}$. Given $z_0^{\star}$, the corresponding integer antenna allocation for problem $({\rm P}1)$ is obtained as $N_b^{\star}(z_0^{\star})$, $b\in\mathcal{B}$. The overall procedure for solving problem $({\rm P}1)$ is summarized in Algorithm~\ref{alg:soln}.

With a direct implementation, the greedy integer allocation for each given $z_0$ has complexity $O(BN)$, since each antenna assignment requires evaluating the marginal gains of all $B$ sectors. Therefore, the overall complexity of Algorithm~\ref{alg:soln} is $O(cBN)$. By maintaining the marginal gains with a max-heap, the complexity of the greedy allocation for each $z_0$ can be reduced to $O(N\log B)$, leading to an overall complexity of $O(cN\log B)$.

\vspace{-1em}
 
\section{Effects of Key Factors on Users' Sum Rate }\label{sec:effectsPara}

In this section, we use the solution obtained by solving problem $({\rm P}1)$ to analyze the effects of several key system factors on the users' maximum sum rate given in \eqref{eq:sumR_objZ}, including the required minimum rate of users $\overline{r}$, the user distribution in terms of the numbers of users in different sectors, i.e., $Q_{b},b\in\mathcal{B},$ and the number of sectors $B$. Note that we omit $z_0$ in the terms involving it in \eqref{eq:sumR_objZ} in this section, since the derived results apply for any given $z_0$, including the optimal $z_0^{\star}$ for problem $({\rm P}1)$. In addition, we assume continuous antenna allocation over sectors for ease of analysis.  

\vspace{-1em}
\subsection{Minimum Rate of Users }

In this subsection, we study the impact of the required minimum rate of users $\overline{r}$ on the solution of problem $({\rm P}1)$ and the resulting users' maximum sum rate. 

From \eqref{sol:antennaAllo}, it follows that, if the following conditions are satisfied:
\begin{equation}\label{eq:con:N_b}
	{Q_{b}}\left( 1+\frac{1}{\nu{\ln2}} \right) - \frac{1}{B\gamma_{0}}\geq N_{b,\min}, \quad \forall b\in\mathcal{B},
\end{equation}
we obtain
\begin{equation} \label{eq:Nsemiclosed}
	N_{b}^{\star} = {Q_{b}}\left( 1+\frac{1}{\nu{\ln2}} \right) - \frac{1}{B\gamma_{0}}, \quad  b\in\mathcal{B}.
\end{equation}
Moreover, from \eqref{eq:nu_constraint} and under the same conditions in \eqref{eq:con:N_b}, we derive a closed-form expression for $\nu$ as 
\begin{equation}\label{eq:nu_closed_form}
	\nu = \frac{K}{\left( N-K+\frac{1}{\gamma_{0}}\right) \ln 2}.
\end{equation}
By substituting $\nu$ in \eqref{eq:nu_closed_form} into \eqref{eq:Nsemiclosed}, we obtain
\begin{equation}\label{eq:optimalN_closed}\
	N_{b}^{\star} =\frac{Q_{b} }{K}\left( N+\frac{1}{\gamma_{0}}\right) - \frac{1}{B\gamma_{0}}, \quad b\in\mathcal{B}. 
\end{equation}
Note that \eqref{eq:optimalN_closed} holds if and only if $N_{b,\min}\leq N_{b}^{\star}$ for all $b$'s. Thus, from \eqref{eq:rho_bar}, we can show  
\begin{equation}\label{eq:con:r}
	\overline{r} \leq \log_{2}\left( \frac{BQ_{b}[(N-K)\gamma_{0}+1]}{K} \right), \quad \forall b\in\mathcal{B},
\end{equation}
which can be equivalently written as 
\begin{equation}
	\overline{r} \leq \log_{2}\left( \frac{BQ_{b,\min}[(N-K)\gamma_{0}+1]}{K} \right) \triangleq \overline{r}_{0}. 
\end{equation}
Accordingly, when ${\overline{r}} \leq \overline{r}_{0}$, the users' maximum sum rate in \eqref{eq:sumR_objZ} is given by 
\begin{equation}\label{eq:sumR_Case1}
	R = \sum_{b=1}^{B} {Q}_{b} \log_{2}\left[B \gamma_{0}\frac{Q_{b} }{K}\left( N-K+\frac{1}{\gamma_{0}}\right)  \right] .
\end{equation}

To draw more useful insights, in the rest of this section, we assume that $\overline{r}$ is sufficiently small such that ${\overline{r}} \leq \overline{r}_{0}$ holds, which is usually the case in practice (e.g., massive MIMO systems) with $N\gg K$, and thus the antenna allocation and sum rate of users are given in closed form in \eqref{eq:optimalN_closed} and \eqref{eq:sumR_Case1}, respectively.

\vspace{-1em}
\subsection{User Distribution }\label{sec:ImpactOfUserDist}

To characterize the effect of user distribution or equivalently the numbers of users over sectors, $Q_b$'s, on the users' maximum sum rate in \eqref{eq:sumR_Case1}, we formulate a new optimization problem to maximize the sum rate given in \eqref{eq:sumR_Case1} by assuming that $Q_b$'s can be freely allocated over sectors subject to the given total number of users, $K$. Note that we assume that the antenna allocation is optimal for any given $Q_b$'s as given in \eqref{eq:optimalN_closed}.  For analytical convenience, we relax $Q_{b},\forall b\in\mathcal{B}$, to continuous values, and thus the problem is formulated as\footnote{In \eqref{eq:sumR_Case1}, the required minimum rate for all users is guaranteed, and thus there is no constraint on the minimum rate of users in problem $({\rm P}4)$. }
\begin{subequations}
	\begin{alignat}{2}
		({\rm P}4):\quad \max_{\boldsymbol{q}\in\mathbb{R}_{+}^{B}} ~ & ~~  \sum_{b=1}^{B}
		Q_{b} \log_{2}\left[ \frac{ Q_{b} B\gamma_{0} }{K} \left( N-K+\frac{1}{\gamma_{0}}\right) \right]  \notag	\\ 
		\mathrm {s.t.}  &~~  \sum_{b=1}^{B} Q_{b} = K, \label{con:append-1} \\ &~~ Q_{b} \geq 1,\forall b\in\mathcal{B}, \label{con:append-2}
	\end{alignat}
\end{subequations}
where $\boldsymbol{q}\triangleq[Q_{1},\cdots,Q_{B}]^{\mathsf{T}}\in\mathbb{R}_{+}^{B}$ denotes the collective vector indicating the number of users in each sector, and \eqref{con:append-2} is required for the condition \eqref{con:1-2} in problem $({\rm P}1)$ to be valid.

The objective function in problem $({\rm P}4)$ is a convex function. Thus, problem $({\rm P}4)$ is a non-convex optimization problem and difficult to solve in general. However, by exploiting the special structure of $({\rm P}4)$, i.e., the objective function consists of identical terms in terms of $Q_b$'s, the optimal solution for problem $({\rm P}4)$ can be derived, which is given by
\begin{equation}\label{eq:opt_Qb_sub}
	\begin{cases}
		Q_{i}^{\star} =& K-B+1, \quad  \exists i\in\mathcal{B}, \\ Q_{j}^{\star} =&1, \quad  \forall j\neq i, j\in\mathcal{B}.
	\end{cases}
\end{equation} 
Thus, the corresponding antenna allocation is given by
\begin{equation}\label{eq:opt_Qb_Nb_sub}
	\begin{cases}
		N_{i}^{\star} =& \frac{K-B+1}{K}\left( N+\frac{1}{\gamma_{0}}\right) - \frac{1}{B\gamma_{0}}, \quad  
		\\ N_{j}^{\star} =& \frac{N-N_{i}^{\star}}{B-1}, \quad  \forall j\neq i, j\in\mathcal{B}.
	\end{cases}
\end{equation}
From \eqref{eq:opt_Qb_sub} and \eqref{eq:opt_Qb_Nb_sub}, we observe that it is optimal to have users located in one sector only as much as possible, and so is the number of antennas allocated to this sector. In practical scenarios, condition \eqref{con:append-2} can be safely omitted since its main purpose is to ensure the feasibility of problem $({\rm P}4)$ from a mathematical standpoint. Thus, we define the most favorable user distribution for attaining the maximum sum rate with the flexible-sector BS as 
\begin{equation}\label{eq:opt_Qb}
	\begin{cases}
		Q_{i} =K, \quad & \exists i\in\mathcal{B}, \\ Q_{j} =0, \quad & \forall j\neq i, j\in\mathcal{B}.
	\end{cases}
\end{equation}
With the optimal user distribution given in \eqref{eq:opt_Qb}, the corresponding optimal antenna allocation is given by
\begin{equation}\label{eq:opt_Qb_Nb}
	\begin{cases}
		N_{i} = N, \quad & 
		 \\ N_{j} =0, \quad & \forall j\neq i, j\in\mathcal{B},
	\end{cases}
\end{equation}
and from \eqref{eq:sumR_obj} the resulting maximum sum rate of all users is given by\footnote{For brevity, we henceforth refer to the worst-case sum rate as the achievable sum rate.}
\begin{equation}\label{eq:sumR_max}
	R_{\max} = K \log_{2}\left[ 1 + B \gamma_{0}  (N-K) \right] .
\end{equation}

Moreover, to obtain more insights, we formulate a problem to minimize the sum rate of all users (instead of maximizing it in $({\rm P}4)$), which is given by
\begin{subequations}
	\begin{alignat}{2}
		({\rm P}5):\quad \min_{\boldsymbol{q}\in\mathbb{R}_{+}^{B}} ~ & ~~  \sum_{b=1}^{B}
		Q_{b} \log_{2}\left[ \frac{ Q_{b} B\gamma_{0} }{K} \left( N-K+\frac{1}{\gamma_{0}}\right) \right]  \notag	\\ 
		\mathrm {s.t.}  &~~  \sum_{b=1}^{B} Q_{b} = K, \label{con:5-1}  \\ &~~ Q_{b}  \geq 1,\forall b\in\mathcal{B}. 
	\end{alignat}
\end{subequations}
Since the objective function in problem $({\rm P}5)$ is convex and the constraints are linear, $({\rm P}5)$ is a convex optimization problem. By applying the Jensen's inequality, we can easily show that in $({\rm P}5)$, the minimum sum rate of all users is attained when the users are equally distributed among all sectors, leading to (assuming $K
\geq B$)
\begin{equation}\label{eq:worst_Qb}
	Q_{b}^{\star} =   \frac{K}{B} ,\quad \forall b\in\mathcal{B}.
\end{equation}
The corresponding antenna allocation is given by
\begin{equation}
	N_{b}^{\star} = \frac{N}{B} , \quad b\in\mathcal{B},
\end{equation}
and the resulting users' minimum sum rate is given by
\begin{equation}\label{eq:sumR_min}
	\begin{split}
		R_{\min} = &~ K\log_{2} \left[1+ B\gamma_{0}\left( \frac{N}{B} - \frac{K}{B}\right)  \right], \\ = &~ K\log_{2} \left[1+ \gamma_{0}(N - K) \right],
	\end{split}
\end{equation}
which is smaller than or equal to $R_{\max}$ in \eqref{eq:sumR_max}, while the equality holds when $B=1$.

Comparing \eqref{eq:sumR_max} and \eqref{eq:sumR_min}, we reveal that the minimum sum rate of all users occurs when users are uniformly distributed across all sectors according to \eqref{eq:worst_Qb}, whereas the maximum sum rate of all users is achieved when users are distributed according to \eqref{eq:opt_Qb}, i.e., all users are located in one sector with no users located in any other sectors (termed as the most favorable user distribution for flexible-sector BS). Thus, it can be inferred that the common rotation of sectors, i.e., $z_{0}^{\star}$, in problem $({\rm P}1)$, should be set to result in unevenly/sparsely distributed users over sectors as much as possible (in the ideal case, approaching to that given in \eqref{eq:opt_Qb}), so as to maximize the users' sum rate. Moreover, in problem $({\rm P}1)$, the sum rate of all users is upper-bounded by $R_{\max}$ given in \eqref{eq:sumR_max} and lower-bounded by $R_{\min}$ given in \eqref{eq:sumR_min}, regardless of the user spatial distributions over angular zones.

\begin{theorem} \label{thm:rateGain}
	From \eqref{eq:sumR_max} and \eqref{eq:sumR_min}, the asymptotic rate gain, defined as the gap between the maximum and minimum achievable rates per user, with the flexible-sector BS under the most favorable user distribution in \eqref{eq:opt_Qb} when $N$ goes to infinity (i.e., $N\gg B$) is given by
	\begin{equation}
		\lim_{N\rightarrow\infty} \left(\frac{R_{\max}}{K} - \frac{R_{\min}}{K} \right) = \log_{2} B,
	\end{equation} 
	in bits-per-second-per-Hertz (bps/Hz). 
\end{theorem}

The rate gain given in Theorem~\ref{thm:rateGain} is explained as follows. When the users are distributed according to \eqref{eq:opt_Qb} within one sector only, the flexible-sector BS can allocate all the antennas to cover this sector by properly setting the sector rotation and moving all antennas from the other sectors to this sector. This results in a sector array gain after applying the ZF combining, which is equal to $(N-K)$ shown in \eqref{eq:sumR_max}. In contrast, when the users are equally distributed over a set of $B$ sectors given in \eqref{eq:worst_Qb}, the flexible-sector BS just needs to properly set its sector rotation (if needed) to cover these sectors, while the antennas are equally allocated to all sectors (similar to the conventional fixed-sector BS with equal antenna allocation over sectors serving users evenly distributed over sectors). This, however, results in a reduced array gain after ZF combining, which is equal to $(N/B-K/B)=(N-K)/B$ given in \eqref{eq:sumR_min}. Thus, the rate gain of the flexible-sector BS over the conventional fixed-sector BS critically depends on the user spatial distribution. 

\vspace{-1em}
\subsection{Number of Sectors }\label{sec:ImapctOfSectors}

Finally, we analyze the impact of $B$ on the users' maximum sum rate by considering non-uniformly distributed users in this subsection. 
We consider two scenarios under the same user distribution with identical sector rotation $z_{0}$ of the flexible-sector BS to investigate the system performance under different setups of $B$.
\begin{itemize}
	\item \textbf{Scenario I}: There are $B_{0}$ sectors. Let ${Q}_{b}^{(\rm I)}$ denote the number of users in sector $b,b\in\{1,\cdots,B_{0}\}$, for the given sector rotation $z_{0}$. From \eqref{eq:sumR_Case1}, the users' sum rate with optimal antenna allocation is given by
	\begin{equation}\label{eq:sumR_Scenario_I}
		R^{(\rm I)} = \sum_{b=1}^{B_{0}} {Q}_{b}^{(\rm I)} \log_{2}\left[B \gamma_{0}\frac{Q_{b}^{(\rm I)} }{K}\left( N-K+\frac{1}{\gamma_{0}}\right)  \right] .
	\end{equation}
	\item \textbf{Scenario II}: There are $2B_{0}$ sectors. Let ${Q}_{b}^{(\rm II)}$ denote the number of users in sector $b,b\in\{1,\cdots,2B_{0}\}$. Based on \eqref{eq:sumR_Case1}, the users' sum rate is given by
	\begin{equation}\label{eq:sumR_Scenario_II}
		R^{(\rm II)} = \sum_{b=1}^{2B_{0}} {Q}_{b}^{(\rm II)} \log_{2}\left[2B_{0} \gamma_{0}\frac{Q_{b}^{(\rm II)} }{K}\left( N-K+\frac{1}{\gamma_{0}}\right)  \right] .
	\end{equation}
\end{itemize} 
Note that the sector rotation $z_0$ in Scenario II satisfies that for any  $j=2i-1,i\in\{1,\cdots,B_{0}\}$, $Q_{j}^{(\rm II)}+Q_{j+1}^{(\rm II)}=Q_{i}^{(\rm I)},\forall i,j$. Then, from \eqref{eq:sumR_Scenario_I} and \eqref{eq:sumR_Scenario_II}, it follows that 
\begin{itemize}
	\item If $Q_{j}^{(\rm II)} = Q_{j+1}^{(\rm II)},\forall i,j,$ $R^{(\rm I)}=R^{(\rm II)}$.
	\item Otherwise, if there exists $ i,j$, such that $Q_{j}^{(\rm II)} \neq Q_{j+1}^{(\rm II)}$, then $R^{(\rm I)}<R^{(\rm II)}$, which can be shown similarly as the solution of problem $({\rm P}4)$.
\end{itemize} 
\begin{remark}\label{rmk:effectOfB}
	Based on the above results, it is evident that if the users are non-uniformly distributed, a larger number of sectors generally tend to enhance the sum rate of users. This is expected since both higher directional antenna gain and higher flexibility in antenna allocation over sectors are offered. This enhancement, however, is accompanied by an increase in hardware cost and implementation complexity.
\end{remark}

\vspace{-1em}
\subsection{Comparison with Conventional BS Architectures}

In this subsection, we compare the achievable sum rates of three types of BSs: the non-sectorized BS with $B=1$, the conventional sectorized BS with $B>1$ and $N_{b}=\frac{N}{B},\forall b\in\mathcal{B}$, and the proposed flexible-sector BS with $B>1$ and $N_{b}^{\star}\sim O(Q_{b}), b\in\mathcal{B}$. The achievable sum rates of these BS architectures are given by
\begin{itemize}
	\item[a)] Non-sectorized BS:
	\begin{equation} \label{eq:perf_non-sectorized}
		R^{(\rm n)} = K\log_{2}[1+\gamma_{0}(N-K)].
	\end{equation}
	\item[b)] Conventional sectorized BS:\footnote{Assuming that $Q_{b}<\frac{N}{B},\forall b\in\mathcal{B}.$}
	\begin{equation}
		R^{(\rm s)} = \sum_{b=1}^{B} Q_{b}\log_{2}\left[ 1+B\gamma_{0}\left( \frac{N}{B}- Q_{b} \right) \right] .
	\end{equation}
	\item[c)] The proposed flexible-sector BS:
	\begin{equation}
		R^{(\rm f)} = \sum_{b=1}^{B} Q_{b}\log_{2}[1+B\gamma_{0}\left(  {N}_{b}^{\star}(z_{0}^{\star}) - Q_{b}(z_{0}^{\star}) \right) ].
	\end{equation}
\end{itemize}
Compared with the non-sectorized and conventional sectorized BSs, the proposed flexible-sector BS preserves the directional antenna gain $B$ while adaptively redistributing the excess spatial degrees of freedom ${N}_{b}^{\star}(z_{0}^{\star}) - Q_{b}(z_{0}^{\star})$ to the sector-level traffic load.
%

\vspace{-1em}

\section{Numerical Results}\label{sec:numRes}

In this section, we provide numerical results to validate the analytical insights and evaluate the robustness of the proposed flexible-sector BS. We focus on three aspects: the impact of common sector rotation and the number of sectors, the performance gain achieved by jointly optimizing sector rotation and antenna allocation under different angular-domain user distributions, and the robustness of the proposed design to non-ideal sector antenna radiation patterns.

Unless otherwise specified, the cell radius is $D=100$ m, and the angular domain is divided into $Z=30$ zones. The normalized SNR is $\gamma_{0}=0$ dB and the minimum user rate requirement is $\overline{r}=5$ bps/Hz. The proposed flexible-sector BS jointly optimizes the common rotation index and antenna allocation according to Algorithm \ref{alg:soln}. To provide a comprehensive evaluation, we consider the following benchmarks:
\begin{itemize}
	\item Load-aware antenna allocation (Antenna allocation only): The common rotation index is fixed with $z_{0}=1$. Accordingly, the antenna resources are flexibly allocated among sectors, i.e., $N_{b}$'s are optimized based on \eqref{eq:N_b_star}.
	\item Adaptive sectorization (Sector rotation only): The antennas are equally allocated over sectors, i.e., $N_{b}=\frac{N}{B},\forall b\in\mathcal{B}$. The common rotation index $z_{0}$ is optimized via the exhaustive search in Section \ref{sec:optSecRot}.
	\item Conventional fixed-sector architecture: Both the common rotation index $z_{0}$ and the antenna allocation are fixed, with $z_{0}=1$ and $N_{b}=\frac{N}{B},\forall b\in\mathcal{B}$.
\end{itemize}

\vspace{-1em}

\subsection{Validation of Analytical Results}


We first consider two representative angular-domain user distributions, as illustrated in Fig. \ref{fig:UserDists} with an average total of $K=50$ users. In angular zone $z$, the number of users in each realization follows $\mathrm{Poisson}(\lambda_z)$ distribution, and the users are uniformly distributed within the corresponding zone. For user distribution I, ($\lambda_z=3$) for ($z=16,\ldots,25$) and ($\lambda_z=1$) otherwise, resulting in a quasi-uniform distribution. For user distribution II, most of the traffic is concentrated in zones 16–25, with $\lambda_{z}=[4,5,6,8,7,5,4,3,3,3]$, whereas only a small background load is present in zones 1 and 11. This produces a highly clustered angular user distribution. These two distributions allow us to examine how angular-domain traffic heterogeneity affects the benefit of flexible sector reconfiguration.

Fig. \ref{fig:perf_z0} shows the average sum rate versus the common rotation index $z_0$ under the two representative user distributions for different numbers of sectors $B$. The sum rate varies periodically with $z_0$, since sector rotation changes the grouping of angular zones and hence the sector-level user loads $Q_b(z_0)$. The performance variation becomes more pronounced for the clustered distribution, where a hotspot can either be concentrated within one sector or split across adjacent sectors depending on the sector boundaries. Moreover, a larger $B$ generally provides a higher maximum sum rate under non-uniform traffic, which is consistent with Remark \ref{rmk:effectOfB}.

Fig. \ref{fig:comp} compares the proposed flexible-sector BS with the three benchmarks versus the total number of antennas $N$. The proposed design consistently achieves the highest sum rate under both user distributions, demonstrating the benefit of jointly adapting sector boundaries and antenna resources. The antenna-allocation-only scheme provides the second-best performance, indicating that traffic-aware antenna redistribution is the dominant source of gain, while sector rotation provides an additional improvement by creating more favorable sector-level traffic loads. Importantly, the performance gain of the proposed design is substantially larger under the clustered distribution, directly supporting the analysis in Section \ref{sec:ImpactOfUserDist} that the flexible-sector architecture is particularly beneficial under highly non-uniform user distributions.

\begin{figure}[h!]
	\centering
		\begin{minipage}[t]{0.45\linewidth}
			\includegraphics[width=0.9\linewidth]{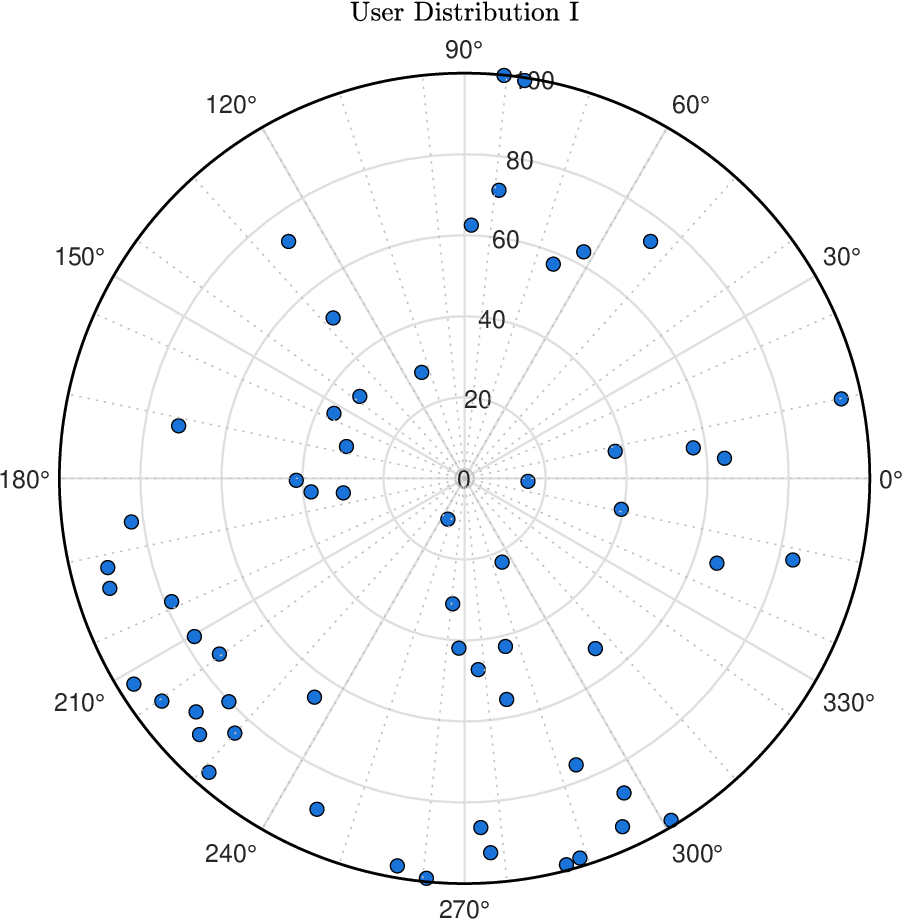} 
			\par\vspace{4pt}  
			\small (a) User distribution I.
			\label{fig:UserDist_TypeI_polar}
		\end{minipage} 
		~~
		\begin{minipage}[t]{0.45\linewidth}
			\includegraphics[width=0.9\linewidth]{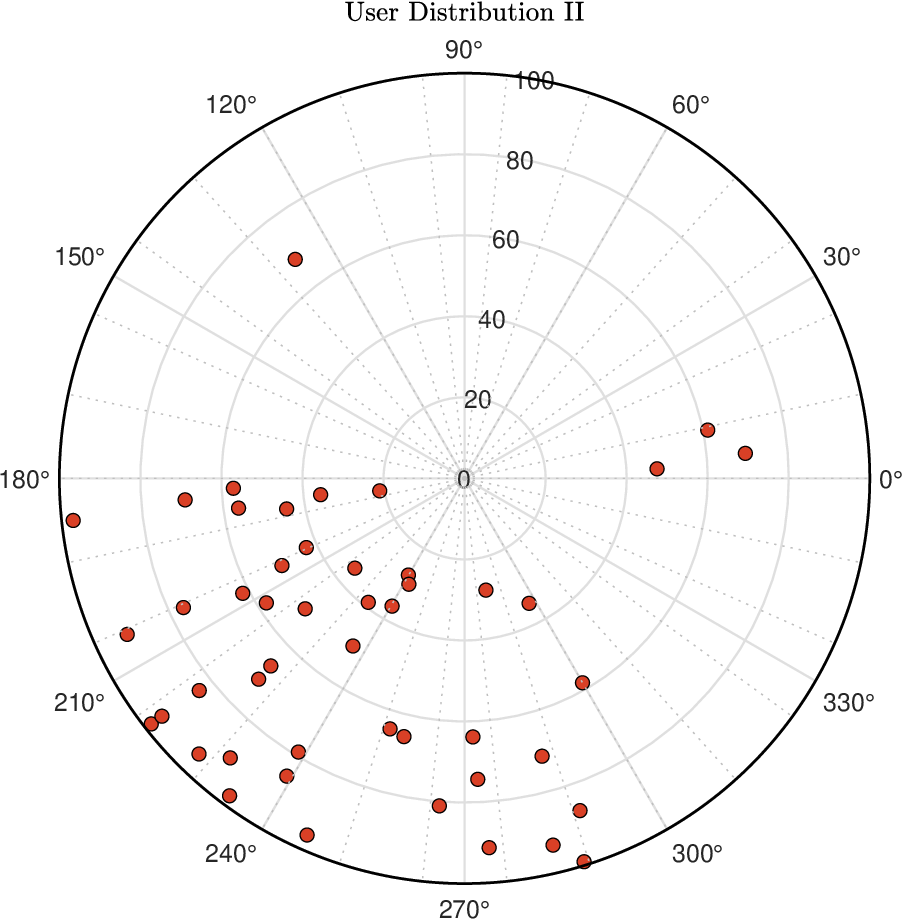}
			\par\vspace{4pt}  
			\small (b) User distribution II.
			\label{fig:UserDist_TypeII_polar}
		\end{minipage} 
		\caption{A realization of the two types of user distributions.}
		\label{fig:UserDists}
	\end{figure}
	
	\begin{figure}[h!]
		\centering 
		\begin{minipage}[t]{0.48\linewidth}
			\includegraphics[width=\linewidth]{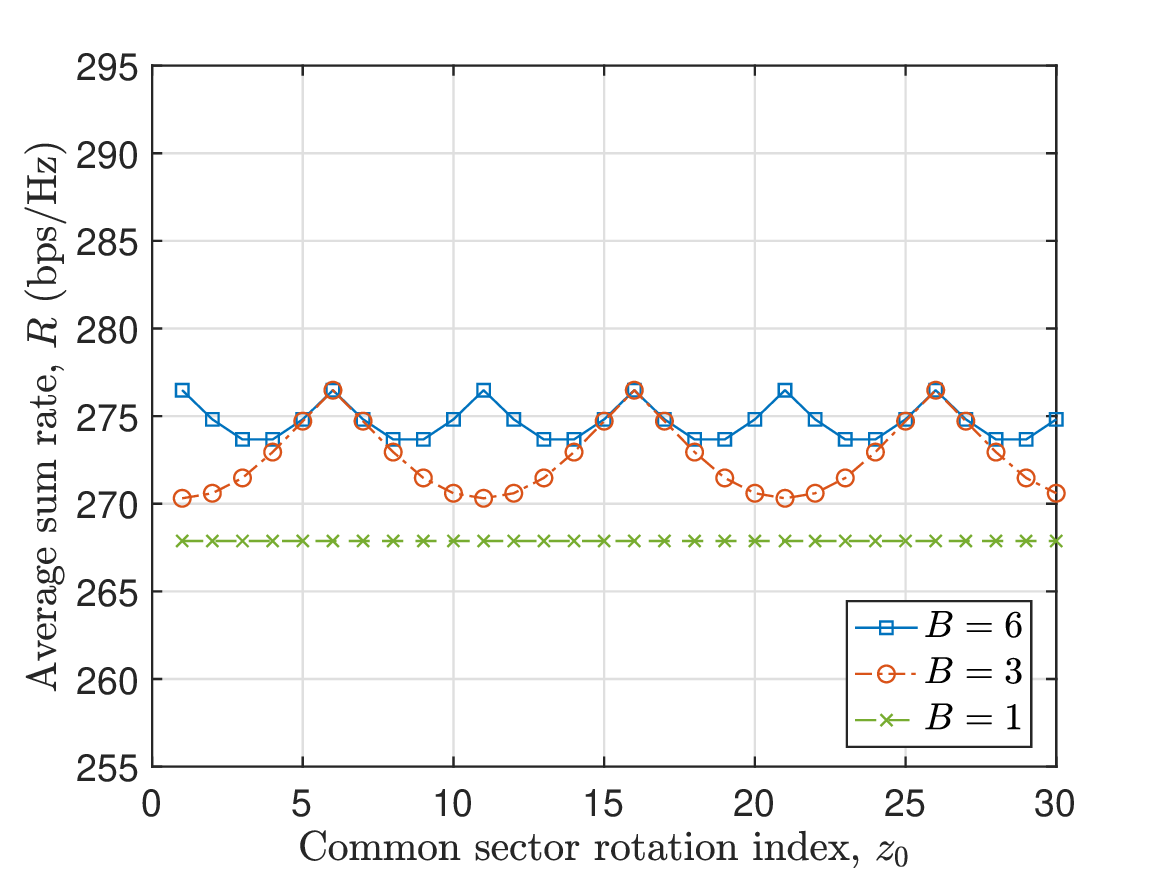} 
			\par\vspace{4pt}  
			\small (a) User distribution I.
			\label{fig:sumR_diff_S_diff_z0_TypeI}
		\end{minipage}  
		\begin{minipage}[t]{0.48\linewidth}
			\includegraphics[width=\linewidth]{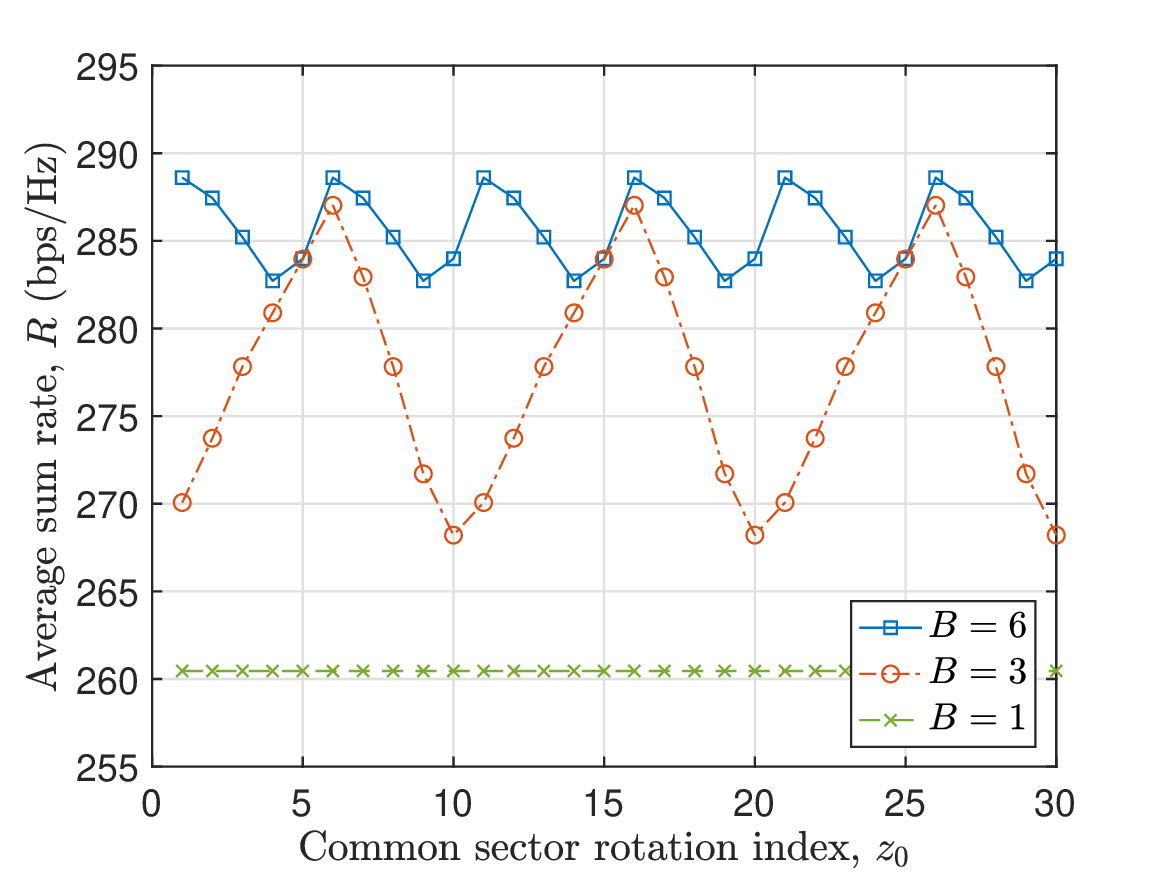}
			\par\vspace{4pt}  
			\small (b) User distribution II.
			\label{fig:sumR_diff_S_diff_z0_TypeIII}
		\end{minipage} 
		\caption{Average sum rate versus $z_{0}$ under different user distributions.}
		\label{fig:perf_z0}
	\end{figure}

	\begin{figure}[h!]
		\centering 
		\begin{minipage}[t]{0.48\linewidth}
			\includegraphics[width=\linewidth]{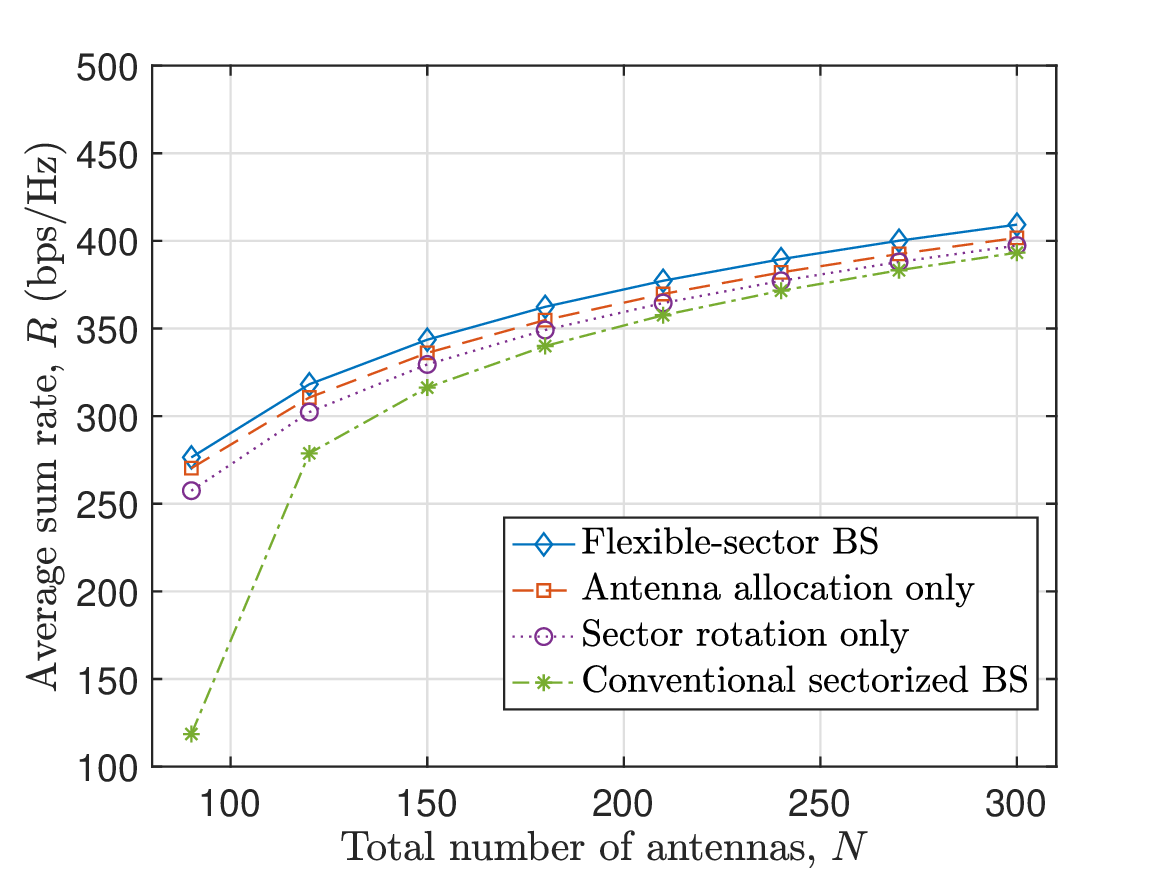} 
			\par\vspace{4pt}  
			\small (a) User distribution I.
			\label{fig:compI}
		\end{minipage}  
		\begin{minipage}[t]{0.48\linewidth}
			\includegraphics[width=\linewidth]{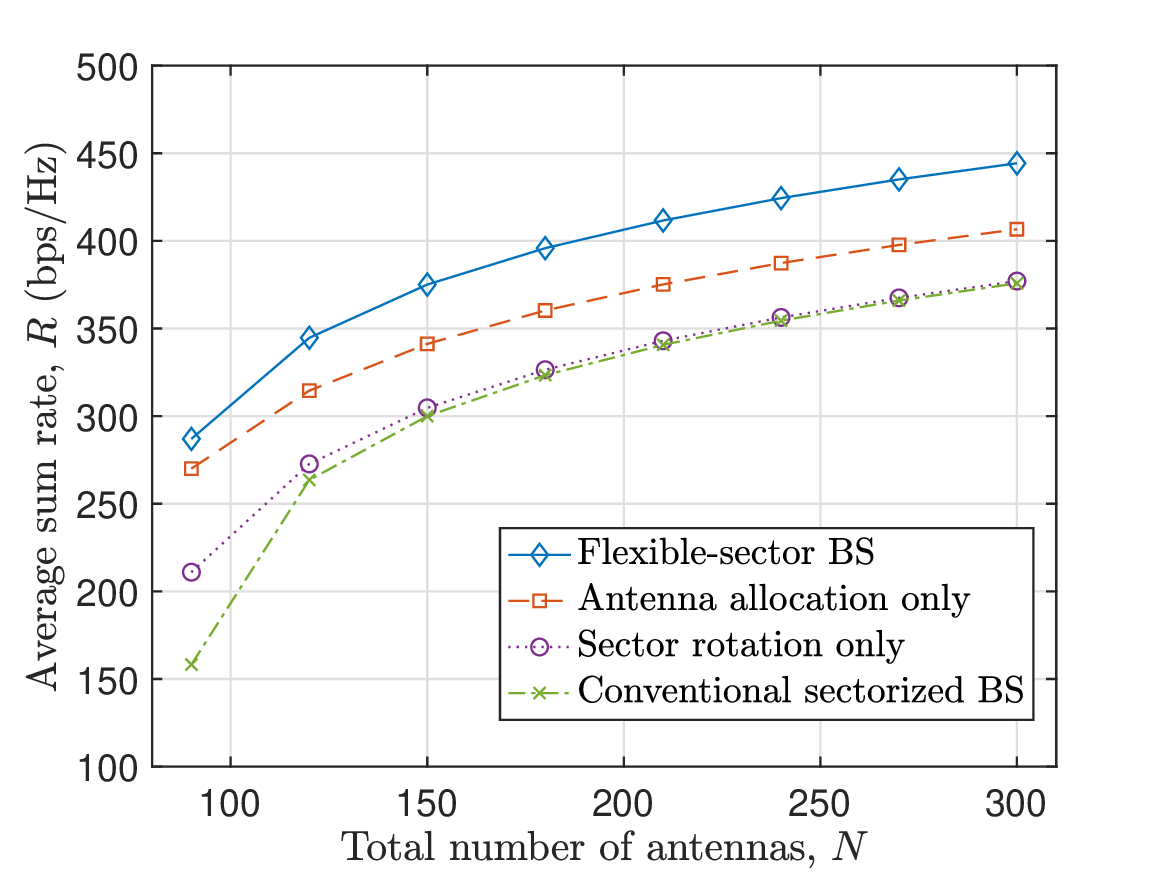}
			\par\vspace{4pt}  
			\small (b) User distribution II.
			\label{fig:compII}
		\end{minipage} 
		\caption{Performance comparison of the flexible-sector BS with other benchmarks.}
		\label{fig:comp}
	\end{figure}

	\begin{figure*}[t!]
		\centering 
		\begin{minipage}[t]{0.24\linewidth}
			\includegraphics[width=\linewidth]{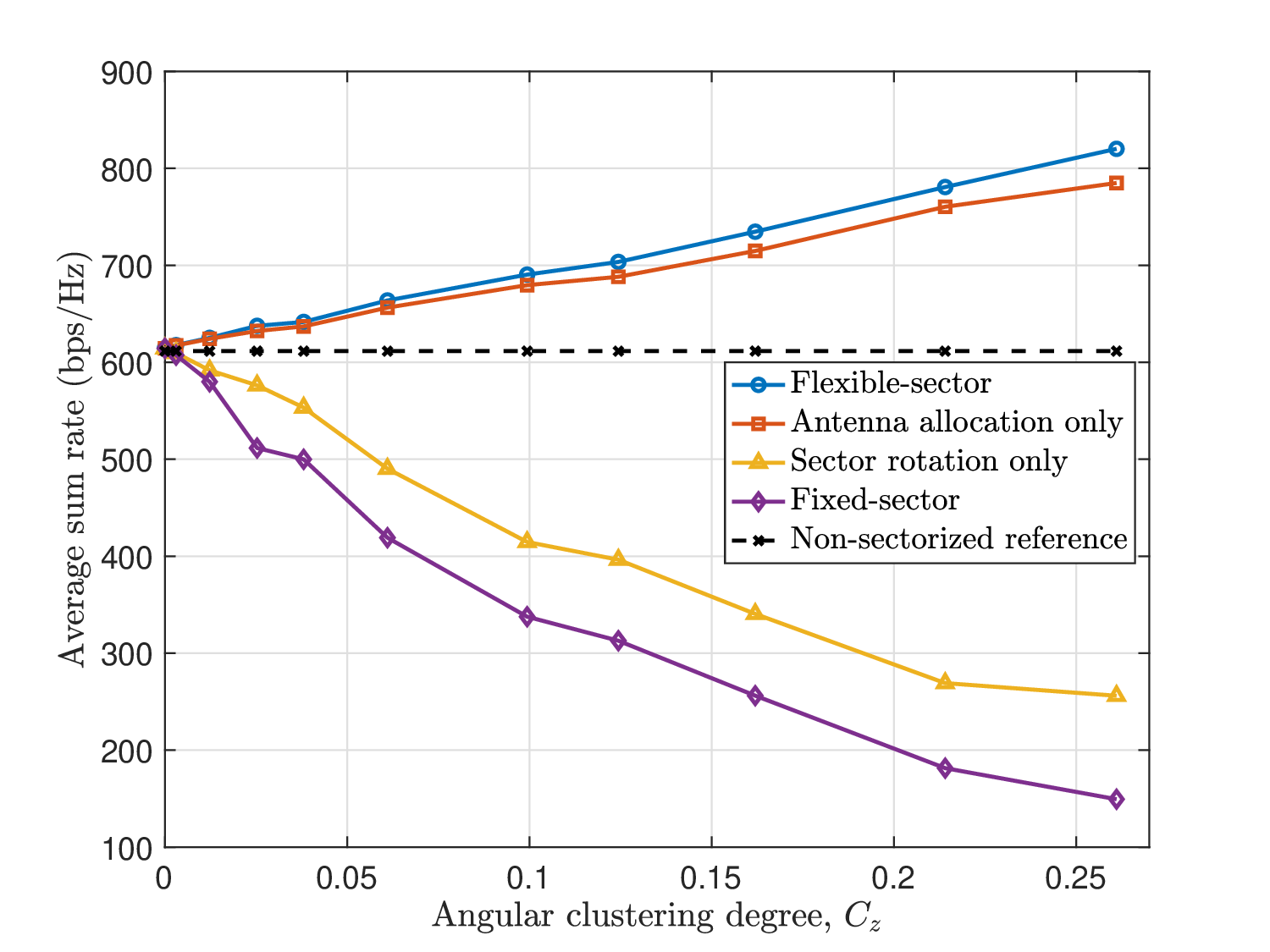} 
			\par\vspace{4pt}  
			\small (a) Ideal sector pattern.
			\label{fig:perf_ideal}
		\end{minipage}  
		\begin{minipage}[t]{0.24\linewidth}
			\includegraphics[width=\linewidth]{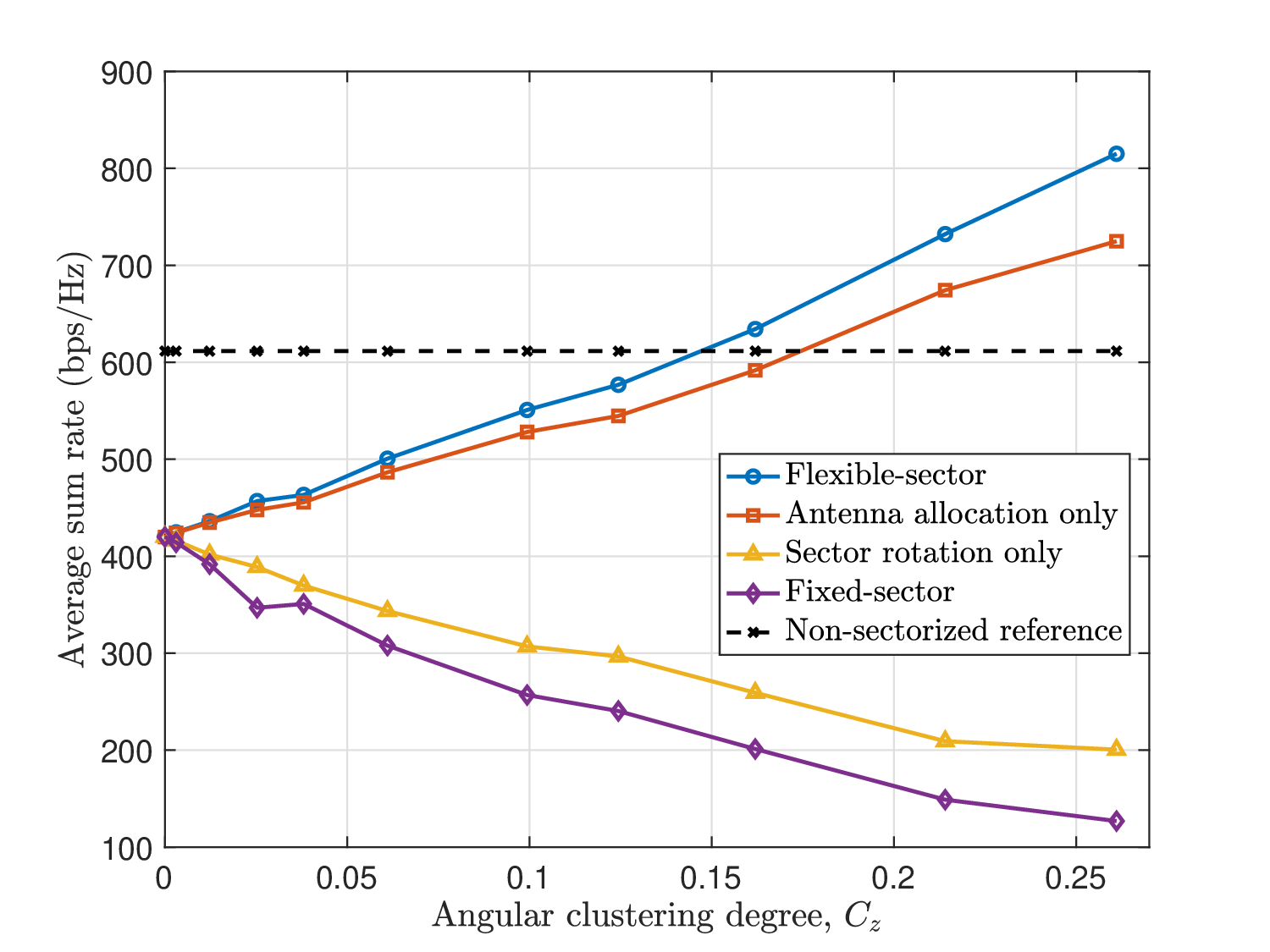}
			\par\vspace{4pt}  
			\small (b) FSL pattern ($A_{\rm sl}=20$ dB).
			\label{fig:perf_A20}
		\end{minipage} 
		\begin{minipage}[t]{0.24\linewidth}
			\includegraphics[width=\linewidth]{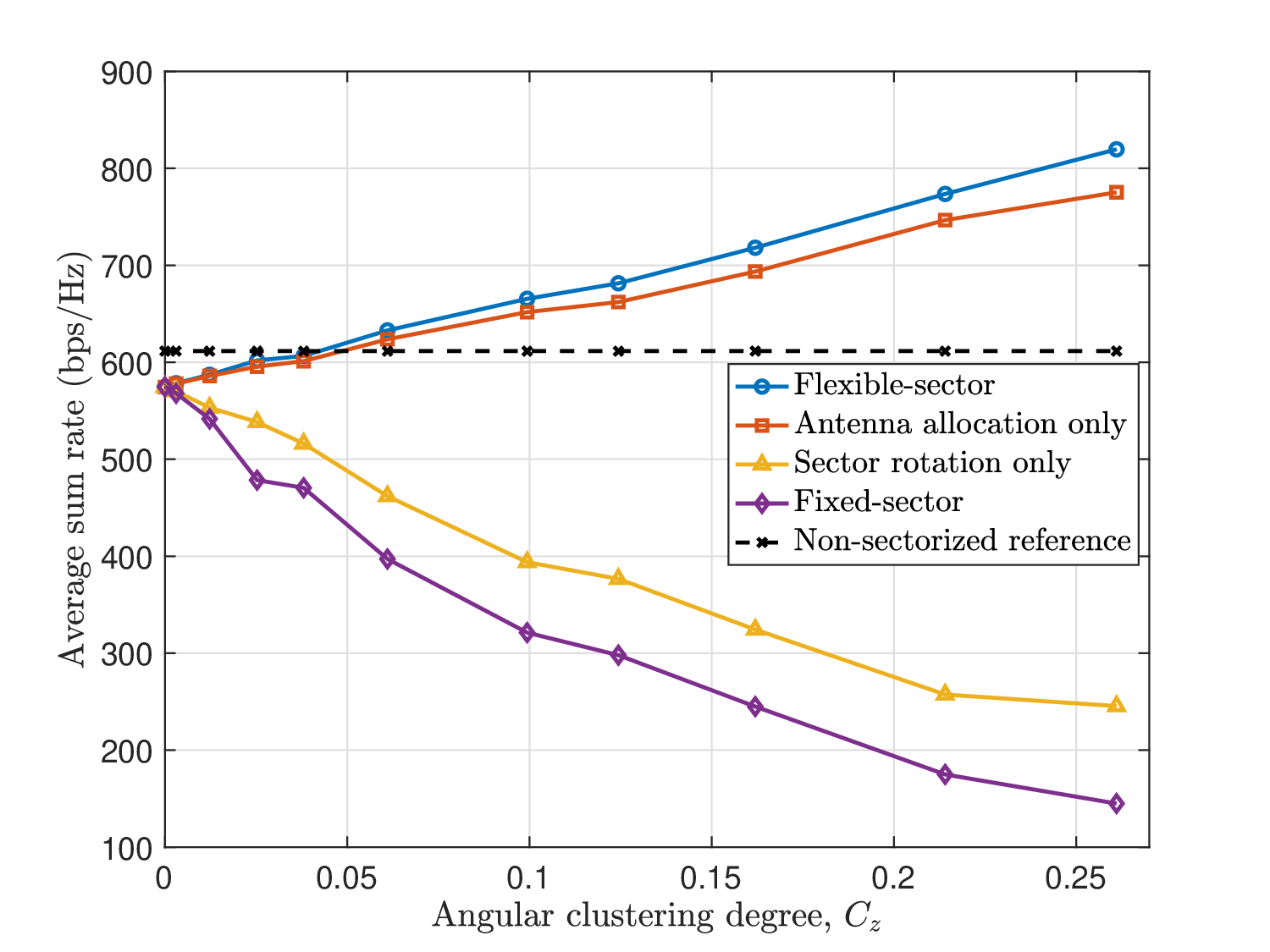}
			\par\vspace{4pt}\small  (c) FSL pattern ($A_{\rm sl}=30$ dB).
			\label{fig:perf_A30}
		\end{minipage} 
		\begin{minipage}[t]{0.24\linewidth}
			\includegraphics[width=\linewidth]{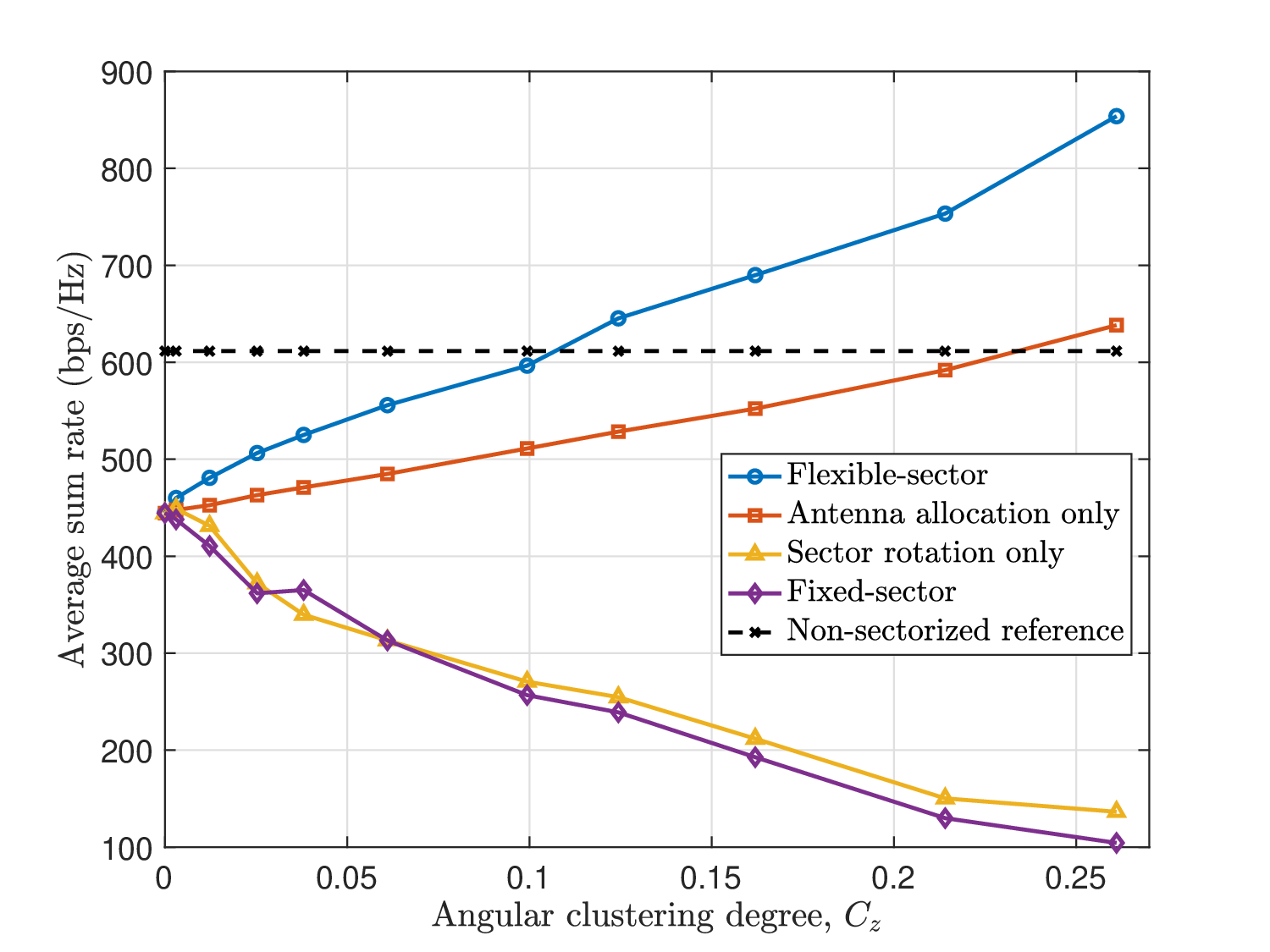}
			\par\vspace{4pt}  
			\small (d) 3GPP-like horizontal pattern.
			\label{fig:perf_3gpps}
		\end{minipage} 
		\caption{Achievable sum rate comparison under different sector antenna radiation patterns.}
		\label{fig:antennaperf}
	\end{figure*}

\vspace{-2em}
\subsection{Robustness to Angular Traffic Clustering and Non-Ideal Antenna Patterns}
\label{sec:numerical_practical}

In this subsection, we evaluate the impact of angular-domain traffic heterogeneity on the performance of the proposed flexible-sector BS under different antenna radiation patterns. The BS is equipped with $B=5$ sectors. Hence, each sector covers $Z/B=6$ angular zones under a fixed sector partition. The total number of BS antennas and users are fixed as $N=200$ and $K=90$, respectively.  
	
To generate user distributions with controllable angular-domain clustering levels, we model the user probability over angular zones as a mixture of a uniform distribution and a circular Gaussian-like hotspot distribution, given by 
\begin{equation}
	p_z(\alpha) = (1-\alpha)\frac{1}{Z} + \alpha q_z,\quad 0\leq \alpha \leq 1,
\end{equation}
where $q_z$ denotes the circular Gaussian-like hotspot distribution centered at a certain angular zone. The parameter $\alpha$ controls the transition from nearly uniform traffic to clustered hotspot traffic. Specifically, $\alpha=0$ corresponds to a uniform angular-domain user distribution, while a larger $\alpha$ leads to a more concentrated user distribution around the hotspot. The integer number of users in each angular zone, denoted by $K_z$, is generated according to $p_z(\alpha)$ while satisfying $\sum_{z=1}^{Z}K_z=K$. To quantify the degree of angular-domain clustering, we use the normalized angular-domain Herfindahl-Hirschman Index (HHI), defined as
\begin{equation}
	C_z=	\frac{\sum_{z=1}^{Z}\left(K_z/K\right)^2-\frac{1}{Z}}{1-\frac{1}{Z}}.
\end{equation}
Hence, $C_z=0$ corresponds to the uniform distribution over all angular zones, while a larger $C_z$ indicates a more spatially concentrated user distribution in the angular domain. For each clustering level, the results are averaged over all possible hotspot center locations to avoid bias caused by a specific alignment between the hotspot and sector boundaries.

We consider the ideal sector pattern in \eqref{eq:antennaG}, the finite-side-lobe (FSL) pattern in \eqref{eq:ASL_antenna} with $A_{\rm sl}=30$ dB and $A_{\rm sl}=20$ dB, and the 3GPP-like horizontal pattern in \eqref{eq:3gpp_antenna}. For the non-ideal antenna patterns, the corresponding SINR of user $k\in\mathcal Q_b$ is computed as
	\begin{equation}
		\tilde \gamma_{b,k}
		=
		\frac{
			P_k|\mathbf w_{b,k}^H \mathbf {\tilde h}_{b,k}|^2
		}{
			\sum\limits_{j\in\mathcal Q_b,j\neq k}
			P_j|\mathbf w_{b,k}^H\mathbf {\tilde h}_{b,j}|^2
			+
			\sum\limits_{\ell\notin\mathcal Q_b}
			P_\ell|\mathbf w_{b,k}^H\mathbf {\tilde h}_{b,\ell}|^2
			+
			\delta^2\|\mathbf w_{b,k}\|^2
		}. \notag
	\end{equation}
	The first interference term is eliminated when ZF combining is applied within the intended sector, whereas the second term captures the residual inter-sector leakage caused by finite side lobes. The corresponding sum rate is obtained by Monte Carlo averaging as
	\begin{equation}
		\tilde R
		=
		\sum_{b=1}^{B}
		\sum_{k\in\mathcal Q_b}
		\mathbb E
		\left[
		\log_2(1+\tilde \gamma_{b,k})
		\right].
	\end{equation}
In addition, we include the non-sectorized BS in \eqref{eq:perf_non-sectorized} as a centralized reference, where all $N$ antennas jointly serve all users without sector partition. Unlike the sectorized architectures, this reference is unaffected by inter-sector leakage and is therefore used only to provide a centralized performance baseline.

Fig. \ref{fig:antennaperf} shows the average achievable sum rate versus the angular clustering degree $C_z$ under four types of antenna patterns, where $A_{\rm m }=30$ dB and $\theta_{\rm 3dB}=0.5\Phi$. First, the proposed flexible-sector BS consistently achieves the highest sum rate among all sectorized architectures under the considered antenna patterns. More importantly, its performance advantage generally increases with $C_z$. As users become increasingly concentrated within a limited angular region, the sector-level user loads become more heterogeneous, creating greater room for traffic-aware antenna redistribution. The proposed design can further rotate the sector boundaries to form more favorable sector-level user groupings and subsequently allocate more antennas to the heavily loaded sectors.

Second, the comparison among the sectorized schemes reveals the distinct roles of sector rotation and antenna allocation. The conventional fixed-sector BS cannot adapt either the sector boundaries or antenna resources and thus suffers significantly under highly clustered user distributions. Sector rotation alone changes the user grouping across sectors, but its equal antenna allocation limits its ability to support heavily loaded sectors. Antenna allocation alone provides a larger performance improvement by redistributing antennas according to the sector-level user loads, whereas the proposed flexible-sector BS achieves an additional gain by jointly adapting both the sector partition and antenna allocation. This observation is consistent with the analytical results in Section \ref{sec:ImpactOfUserDist}.

Third, non-ideal antenna radiation reduces the sector isolation but does not eliminate the benefit of flexible reconfiguration. The FSL pattern with $A_{\rm sl}=30$ dB yields results close to those under the ideal sector pattern. With $A_{\rm sl}=20$ dB, stronger inter-sector leakage reduces the achievable rates of the sectorized architectures, while the proposed flexible-sector BS still maintains a clear advantage, particularly at moderate and large $C_z$. A similar qualitative behavior is observed under the 3GPP-like horizontal pattern, confirming that the proposed design remains effective with a smooth and non-ideal azimuth radiation response.

The non-sectorized reference exhibits a clustering-insensitive sum rate since neither its antenna configuration nor its receiver architecture depends on the angular user distribution.

Consequently, it can outperform the sectorized architectures under strongly non-ideal radiation patterns when the users are nearly uniformly distributed. As $C_z$ increases, however, the proposed flexible-sector BS increasingly exploits the angular-domain user heterogeneity and eventually surpasses the non-sectorized reference. These results demonstrate that the proposed architecture is particularly attractive for hotspot scenarios and its gain is robust to practical sector antenna radiation patterns.

Fig. \ref{fig:diff_Rx} shows the average achievable sum rate under the 3GPP-like antenna pattern with different linear receivers, where $\rho$ denotes the regularization factor of the RZF receiver and the Rician factor is set to $6$ dB. For both Rayleigh and Rician fading channels, the proposed flexible-sector BS achieves evident gains over the benchmark schemes when RZF/LMMSE receivers are employed, and the gain generally increases with the angular clustering degree $C_z$. This is because a larger $C_z$ indicates stronger angular-domain traffic heterogeneity, which can be better exploited by jointly optimizing sector rotation and antenna allocation. It is also observed that the gain of the plain ZF receiver is less pronounced under the non-ideal 3GPP-like pattern due to residual inter-sector leakage. Overall, these results demonstrate that the proposed traffic-aware flexible-sector design remains effective under practical antenna patterns and different linear receiver designs.

\begin{figure}[h!]
	\centering 
	\begin{minipage}[t]{0.48\linewidth}
		\includegraphics[width=\linewidth]{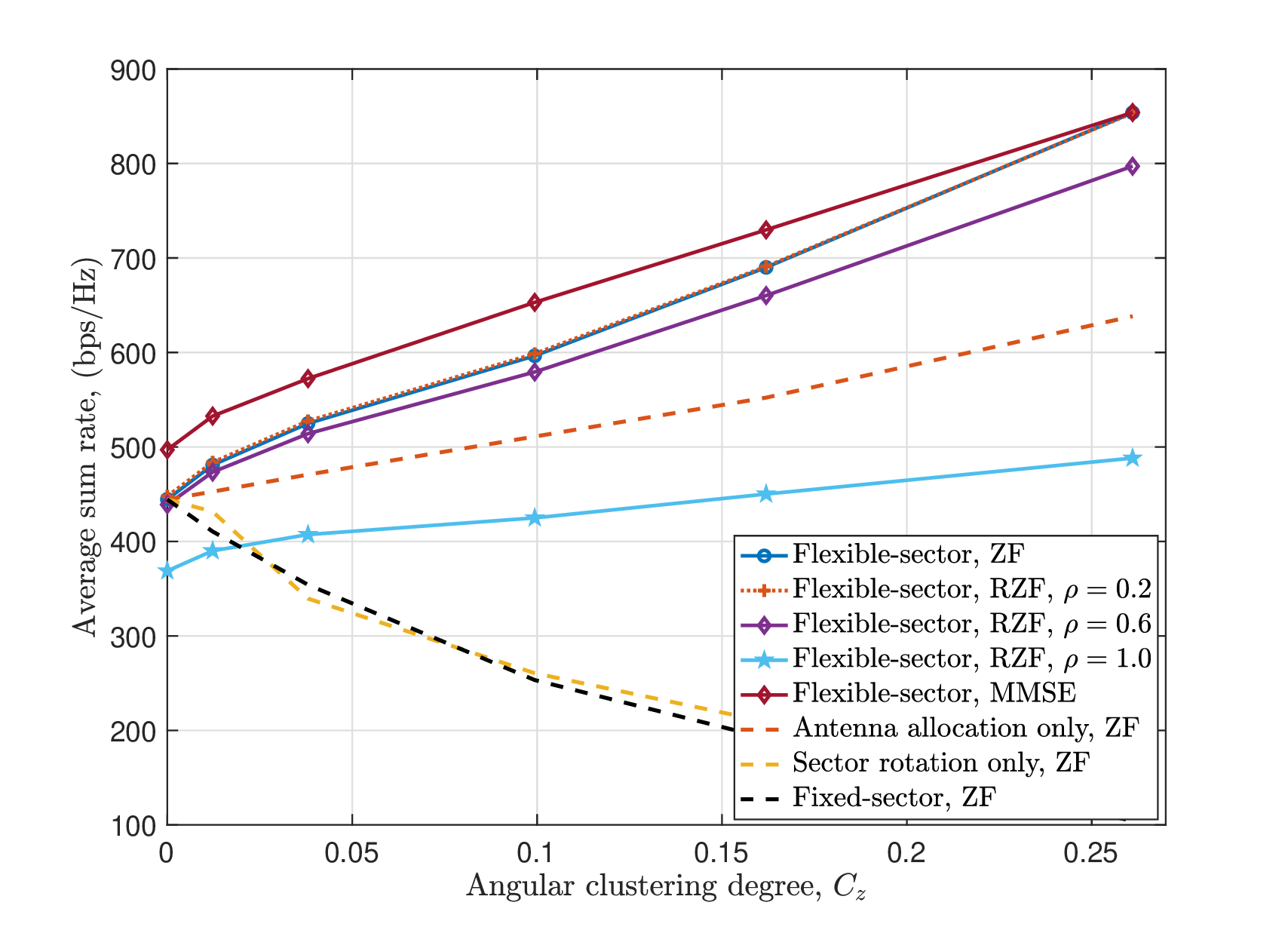} 
		\par\vspace{4pt}  
		\small (a) Rayleigh.
		\label{fig:rayleigh}
	\end{minipage}  
	\begin{minipage}[t]{0.48\linewidth}
		\includegraphics[width=\linewidth]{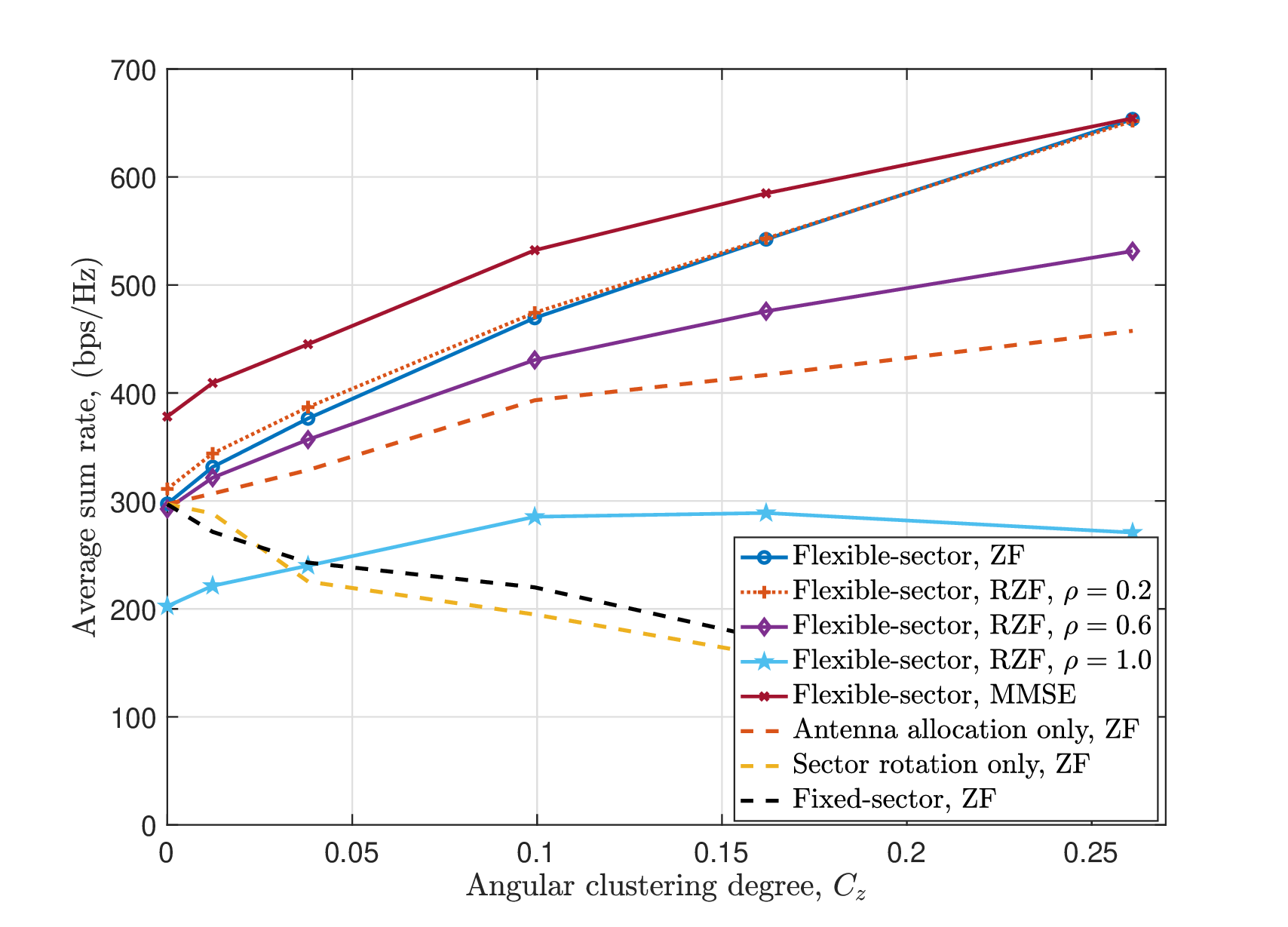}
		\par\vspace{4pt}  
		\small (b) Rician.
		\label{fig:rician}
	\end{minipage} 
	\caption{Average achievable sum rate versus angular clustering degree $C_z$ under the 3GPP-like antenna pattern with different linear receivers.}
	\label{fig:diff_Rx}
\end{figure}

%

\vspace{-2em}
\section{Conclusion}\label{sec:conclusion}


This paper proposed a traffic-aware flexible-sector BS architecture that combines common sector rotation with mechanically reconfigurable antenna-module pooling. Unlike general movable-antenna systems that optimize fine-grained antenna positions, the proposed design reallocates independently processed receive modules across sectors according to long-term angular traffic statistics. We developed an angular-domain traffic model, derived a ZF-based sum-rate lower bound, and jointly optimized sector rotation and integer antenna module allocation. The proposed marginal-gain algorithm achieves the optimal integer antenna allocation for each rotation candidate. Numerical results showed that the architecture is particularly effective under clustered traffic and remains robust to non-ideal antenna patterns and different linear receivers. Future work will consider inter-cell interference, practical reconfiguration constraints, and three-dimensional sector and traffic models.

\appendices

\vspace{-1em}

\section{Proof of Lemma~\ref{lem:bounds}}\label{append:Lemma1}

It can be shown that given $a, b > 0$, $f(x) = \log_{2}(1 + ax)$ is concave over $x > 0$, and $g(x) = \log_{2}(1+b/x)$ is convex over $x > 0$. As a result, according to the Jensen’s inequality, ${r}_{b}^{(\rm l)}$ given in \eqref{eq:lowerB_r_b} and ${r}_{b}^{(\rm u)}$ given in \eqref{eq:upperB_r_b} serve as the lower bound and upper bound of ${r}_{b}$, respectively \cite{liu2019comp}. Based on the orthogonal property of ZF combining, the combining vector $\boldsymbol{w}_{k}$ for user $k$ in sector $b$ is orthogonal to the subspace spanned by the channel vectors from the other $Q_{b}-1$ users in the same sector. Then, the upper bound ${r}_{b}^{(\rm u)}$ given in \eqref{eq:upperB_r_b} can be obtained, and the detailed proof is similar to that of \cite[Theorem 2]{liu2019comp}, and thus is omitted. Note that since the wireless channels between users and the BS follow the Rayleigh fading, i.e., $\boldsymbol{g}_{k}\sim\mathcal{CN}(\boldsymbol{0},\boldsymbol{I}_{N_{b}})$, $\boldsymbol{G}_{b}^{\mathsf{H}}\boldsymbol{G}_{b}$ is a Wishart matrix \cite{tulino2004random}. According to \cite[Lemma 2.10]{tulino2004random}, we have
	$\mathbb{E}[\text{tr}((\boldsymbol{G}_{b}^{\mathsf{H}}\boldsymbol{G}_{b})^{-1})] = \frac{Q_{b}}{N_{b}-Q_{b}}$.
Therefore, the lower bound ${r}_{b}^{(\rm l)}$ given in \eqref{eq:lowerB_r_b}  can be obtained. This completes the proof.
 
\vspace{-1em}

\section{Derivation of \eqref{sol:antennaAllo} and \eqref{eq:nu_constraint}}\label{append:Lagrange}

By applying the Lagrange multipliers to problem $({\rm P}3)$, we obtain the Lagrangian of $({\rm P}3)$ as 
\begin{align}
	L(\boldsymbol{n},\boldsymbol{\mu},\nu) = &~ \sum_{b=1}^{B} Q_{b}\log_{2}\left[ 1+ B\gamma_{0}(N_{b}-Q_{b}) \right]  \notag \\ &~ + \boldsymbol{\mu}^{\mathsf{T}}( \boldsymbol{n}- \boldsymbol{n}_{\min}) - \nu\left(\sum_{b=1}^{B}N_{b} - N \right), 
\end{align}
where $\boldsymbol{n}_{\min}\triangleq\{N_{1,\min},\cdots,N_{B,\min}\}$, $\boldsymbol{\mu}\triangleq\{\mu_{1},\cdots,\mu_{B}\}$, $\nu\geq0$, and ${\mu}_{b}\geq0$ are the Lagrange multipliers associated with constraints \eqref{con:3-1} and \eqref{con:3-2} in problem $({\rm P}3)$, respectively. 
Thus, by applying the Karush-Kuhn-Tucker (KKT) conditions, 
 we determine $N_{b}^{\star}$ in two cases:
	If ${\mu}_{b} >0$ and 
		$\nu \geq \frac{Q_{b}}{\ln2}\left( N_{b,\min}-Q_{b}+\frac{1}{B\gamma_{0}} \right)^{-1}$,
	we have
		$N_{b}^{\star}=N_{b,\min}, \quad b\in\mathcal{B}$. 
	If ${\mu}_{b}=0$
	and
		$\nu < \frac{Q_{b}}{\ln2}\left( N_{b,\min}-Q_{b}+\frac{1}{B\gamma_{0}} \right)^{-1}$,
	we have 
		$N_{b}^{\star}= Q_{b}\left( 1+\frac{1}{ {\nu}^{\star}\ln2} \right) - \frac{1}{B\gamma_{0}},\quad b\in\mathcal{B}$. 
Therefore, the solution for problem $({\rm P}3)$ is given by \eqref{sol:antennaAllo}.

\vspace{-1em}

\section{Solution of Problem $({\rm P}5)$ }\label{append:UserDist}

By applying the Lagrange multiplier, we have
\begin{equation}\label{append:eq:Lag_Q}
	L(\boldsymbol{q},\xi) = \sum_{b=1}^{B}
	Q_{b} \log_{2}\left( Q_{b}\eta \right)  + \xi \left( \sum_{b=1}^{B} Q_{b} - K \right) , 
\end{equation}
where $\eta= \frac{B\gamma_{0}(N-K+\frac{1}{\gamma_{0}})}{K}$, and $\xi$ is the Lagrange multiplier associated with constraint \eqref{con:5-1}. Then, by solving the KKT conditions of \eqref{append:eq:Lag_Q}, we obtain
	$Q_{b}^{\star} =  2^{-\xi-\frac{1}{\ln2}}\frac{K}{B\gamma_{0}(N+\frac{1}{\gamma_{0}}-1)}  ,\quad \forall b\in\mathcal{B}$, which is a constant.
As such, the optimal solution for problem $({\rm P}5)$ corresponds to the case when the users are equally distributed among sectors.

\bibliographystyle{IEEEtran}  
\bibliography{references}

\end{document}